\pdfoutput=1
\RequirePackage{latexml}

\PassOptionsToPackage{
	pagebackref
}{hyperref}

\iflatexml
	\documentclass{article}
	\usepackage{natbib}
	\usepackage[colorlinks]{hyperref}
	\usepackage{amsmath,amssymb,amsthm,tikz}
	\newtheorem{theorem}{Theorem}[section]
	
	\newtheorem{proposition}[theorem]{Proposition}
	\newtheorem{corollary}[theorem]{Corollary}
	\newtheorem{definition}[theorem]{Definition}
	\newtheorem{example}[theorem]{Example}
	
\else
	
	\documentclass[format=acmsmall, nonacm=true, screen]{acmart}

	\acmConference{}{}{}
	\renewcommand{\footnotetextcopyrightpermission}[1]{} %
	\renewcommand{\footnotetextauthorsaddresses}[1]{} %
	\makeatletter
	\def\@ACM@checkaffil{%
		\if@ACM@instpresent\else
		\ClassWarningNoLine{\@classname}{No institution present for an affiliation}%
		\fi
		\if@ACM@citypresent\else
		\ClassWarningNoLine{\@classname}{No city present for an affiliation}%
		\fi
		\if@ACM@countrypresent\else
		\ClassWarningNoLine{\@classname}{No country present for an affiliation}%
		\fi
	}
	\makeatother
	\setcitestyle{authoryear}
	
	\usepackage{wrapstuff}
\fi

\usepackage[utf8]{inputenc}
\usepackage{microtype}

\usepackage{tikz}
\usepackage[export]{adjustbox} %
\usepackage[nameinlink]{cleveref}
\AddToHook{cmd/appendix/before}{\crefalias{section}{appendix}\crefalias{subsection}{appendix}}
\usepackage{booktabs}
\usepackage{caption}
\usepackage{subcaption}
\newcommand{\best}{\textsf{best}}
\newcommand{\supp}{\textsf{supp}}
\newcommand{\rep}{\textsf{score}}
\newcommand{\share}{\textsf{share}}
\newcommand{\score}{\textsf{score}}
\newcommand{\DO}{\textup{DO}}
\newcommand{\GP}{\textup{GP}}
\newcommand{\STV}{\textup{STV}}
\newcommand{\cmark}{%
	\tikz[scale=0.23,draw=green!50!black] {
		\draw[line width=0.7,line cap=round] (0.25,0) to [bend left=10] (1,1);
		\draw[line width=0.8,line cap=round] (0,0.35) to [bend right=1] (0.23,0);
}}
\newcommand{\xmark}{%
	\tikz[scale=0.23,draw=black!50!red] {
		\draw[line width=0.7,line cap=round] (0,0) to [bend left=6] (1,1);
		\draw[line width=0.7,line cap=round] (0.2,0.95) to [bend right=3] (0.8,0.05);
}}

\title{Reallocating Wasted Votes in Proportional Parliamentary Elections with Thresholds}

\iflatexml
	\author{Théo Delemazure \\ CNRS, LAMSADE, Université Paris Dauphine - PSL
		\and
		Rupert Freeman \\ Darden School of Business
		\and
		Jérôme Lang \\ CNRS, LAMSADE, Université Paris Dauphine - PSL
		\and
		Jean-François Laslier \\ CNRS, Paris School of Economics
		\and
		Dominik Peters \\ CNRS, LAMSADE, Université Paris Dauphine - PSL
	}
\else
	\author{Théo Delemazure}
	\email{theo.delemazure@dauphine.eu}
	\affiliation{\institution{CNRS, LAMSADE, Université Paris Dauphine - PSL} \country{France}}
	\author{Rupert Freeman}
	\email{FreemanR@darden.virginia.edu}
	\affiliation{\institution{Darden School of Business} \country{USA}}
	\author{Jérôme Lang}
	\email{jerome.lang@lamsade.dauphine.fr}
	\affiliation{\institution{CNRS, LAMSADE, Université Paris Dauphine - PSL} \country{France}}
	\author{Jean-François Laslier}
	\email{jean-francois.laslier@ens.fr}
	\affiliation{\institution{CNRS, Paris School of Economics} \country{France}}
	\author{Dominik Peters}
	\email{dominik.peters@lamsade.dauphine.fr}
	\affiliation{\institution{CNRS, LAMSADE, Université Paris Dauphine - PSL} \country{France}}
\fi

\begin{abstract}
	In many proportional parliamentary elections, electoral thresholds (typically 3-5\%) are used to promote stability and governability by preventing the election of parties with very small representation. However, these thresholds often result in a significant number of ``wasted votes'' cast for parties that fail to meet the threshold, which reduces representativeness. One proposal is to allow voters to specify replacement votes, by either indicating a second choice party or by ranking a subset of the parties, but there are several ways of deciding on the scores of the parties (and thus the composition of the parliament) given those votes. We introduce a formal model of party voting with thresholds, and compare a variety of party selection rules axiomatically, and experimentally using a dataset we collected during the 2024 European election in France. We identify three particularly attractive rules, called Direct Winners Only (DO), Single Transferable Vote (STV) and Greedy Plurality (GP).
\end{abstract}

\begin{document}

\addtocontents{toc}{\protect\setcounter{tocdepth}{-1}} %
\maketitle

\iflatexml\else
\vspace{1cm}
\tableofcontents
\addtocontents{toc}{\protect\setcounter{tocdepth}{2}}
\fi

\section{Introduction}
\label{sec:intro}

Many democracies elect their parliaments using proportional representation, typically implemented using party lists of candidates, with each voter voting for one party. Most countries impose an \emph{electoral threshold}, a minimum percentage of votes (usually between $0$ and $6\%$) that is necessary for a party to enter parliament \citep{Farrell2011,pukelsheim2014proportional}. Lists that do not gather the required votes get no seats, and the votes for those lists are ``lost'' or ``wasted'' and not used to distribute seats in parliament. In this paper, we will explore ways to prevent this phenomenon of wasted votes.

Not all jurisdictions that use proportional representation impose a threshold. While no votes are lost in such a system, it comes with the risk of having a fragmented parliament with many parties, which makes forming and maintaining a governing coalition difficult. Infamously, the Weimar Republic (1918--1933) did not use a threshold and saw the number of parties present in the \emph{Reichstag} steadily grow to up to 15 in 1930. This led to political chaos: over 13 years there were 16 governments (only five of which had a majority) and 8 elections. The lack of a threshold and the resulting popular dissatisfaction with the political system are widely seen as one contributing factor to the rise of National Socialism \citep{falter2020hitlers}, though the magnitude of its influence is disputed \citep{antoni1980weimar}. Citing this experience, post-war Germany instituted a 5\% threshold in 1953.

\iflatexml
\begin{figure}
	\includegraphics[width=5.8cm]{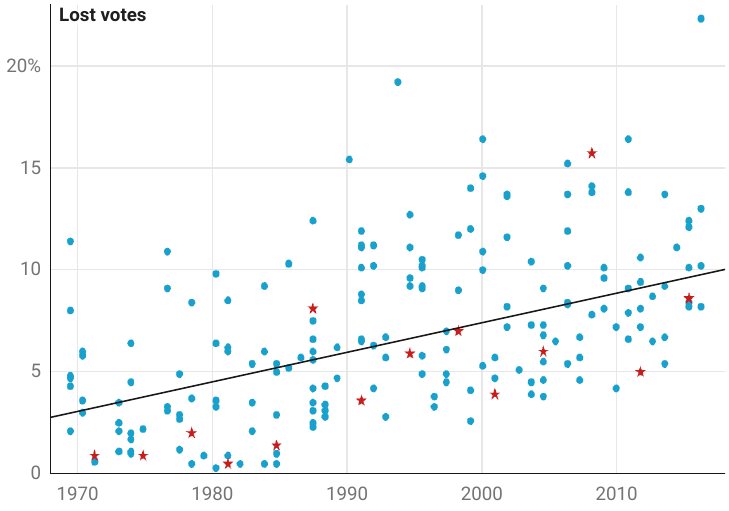}
	\vspace{-8pt}
	\caption{The fraction of lost votes in German federal and state elections (shown as red stars and blue circles, respectively) between 1970 and 2022. Reproduced from \citet{benken2023differenziert}.}
	\label{fig:germany-trend}
\end{figure}
\else
\begin{wrapstuff}[r,type=figure,width=6cm]
	\includegraphics[width=5.8cm]{figs/lost_votes_in_germany.pdf}
	\vspace{-8pt}
	\caption{The fraction of lost votes in German federal and state elections (shown as red stars and blue circles, respectively) between 1970 and 2022. Reproduced from \citet{benken2023differenziert}.}
	\label{fig:germany-trend}
\end{wrapstuff}
\fi
While thresholds limit the number of parties in parliament and improve governability, they also have drawbacks. They can lead to a significant number of wasted votes, which violates the principles of proportionality and equality of votes. \citet{benken2023differenziert} has cataloged the fraction of votes that were lost in German elections since 1970 and finds a steady upward trend (\Cref{fig:germany-trend}). His dataset includes the 2022 election in the state of Saarland, where a record of 22\% of votes were lost. Another example is the 2019 election for the French members of the European Parliament, where the 5\% threshold led to 19.8\% of wasted votes. 
An extreme example is Turkey's 2002 election that used a 10\% threshold (since reduced to 7\%) and where 46\% of votes were lost -- resulting in a party with 34\% of the votes obtaining almost a two-thirds majority in parliament \citep{ozel2003turkey}.
Besides wasting votes and thereby reducing representativeness or even creating false majorities, thresholds also discourage the formation of new parties, hinder the growth of small parties, and require voters to vote strategically \citep{decker2016}.

Thus, na\"ively, there appears to be tradeoff between the problems of a threshold and the risk of political fragmentation. However, there are promising proposals that could alleviate the problems of the threshold without taking away its advantages. In particular, we could elicit additional information from voters regarding their preferences over parties. For example, we could ask voters for a second choice of party. If their first choice of party misses the threshold, their vote is instead counted for the second choice. More generally, we could allow voters to provide a (partial) \emph{ranking} of the parties, and keep redistributing the vote until we reach a party that met the threshold.

This idea has been extensively discussed in Germany under the name ``replacement vote'' (\emph{Ersatzstimme}, sometimes translated into English as ``spare vote''). It appears in the election program of one party for the 2025 German parliament election \citep[page 17]{volt2025}, and laws implementing it have been proposed (but not adopted) in three German states in 2013--15. It is also the main subject of a recent academic edited volume in German language \citep{benken2023mehrdemokratie}. Elsewhere, the idea has been discussed by the Independent Electoral Review in New Zealand, a committee established by the Justice Minister, which noted the strong support that the proposal received during their consultation, though they recommended to instead lower the threshold to not complicate the voting process \citep[numbers 4.34 and 4.58]{nz2023report}.

Despite this broad attention to these proposals, they have never been studied from a social choice perspective (to the best of our knowledge). As we will see, there are many interesting voting theoretic questions not answered by the high-level description of how to process the voters' second choices or rankings.
To study these questions, we introduce a new framework of \emph{party selection rules}, which take as input a profile of (possibly truncated) rankings over parties and a threshold $\tau$ (an absolute number of votes). They output a subset of parties: those that will be included in the parliament. For a given selection of parties, a voter is \emph{represented} by a party $c$ if it is their most-preferred party in the selection. We require that a selection should be \emph{feasible}, in the sense that each selected party represents at least $\tau$ voters. We view such a feasible selection as fully specifying the make-up of a parliament, though in practice we will need to apply an apportionment rule (such as D'Hondt) to determine exactly how many seats each party obtains, as a function of the number of voters it represents.
Because our model allows truncated rankings, it in particular allows for applications where voters can rank at most two parties, which is the most commonly discussed variant in the threshold context.

As has been recognized in the discussion in Germany, there are at least two possible party selection rules \citep[pp. 52--62]{benken2023mehrdemokratie}. The simplest is what we call the \emph{direct winners only} (DO) rule, which selects exactly those parties who are ranked in first place by at least $\tau$ voters, and assigns voters who do not rank any of those parties in first place to their most-preferred selected party. Thus, under DO, there is just a single round of reassignments. Another option reassigns votes in multiple rounds. We call the resulting rule the \emph{single transferable vote} (STV) due to its close similarity to the rule of the same name used in systems that let voters rank candidates instead of parties \citep{Tideman95}. STV works by repeatedly identifying the party with the fewest first-place votes, and eliminates it from the profile. It repeats this until the set of remaining parties is feasible. We add a third party selection rule to the collection: the \emph{greedy plurality} (GP) rule sorts parties by the number of voters placing it in first position, then iteratively adds parties to the selection starting from those with largest score, as long as the addition keeps the selection feasible.

We compare these three rules (DO, STV, GP) using the axiomatic method, by defining a variety of properties appropriate for the party selection model with thresholds. For example, DO is the only of these rules satisfying monotonicity, and we give an axiomatic characterization of DO using a reinforcement-like consistency condition. For STV, we show that it satisfies clone-proofness \citep{tideman1987independence} and represents solid coalitions \citep{dummett1984voting}, and we characterize STV using an independence condition that is quite strong but formalizes a common normative intuition advanced in favor of STV. For GP, we note that it satisfies set-maximality which is a kind of efficiency axiom.

Parliamentary elections with thresholds typically involve a good deal of strategizing on the part of the voters \citep{Abramson}. In particular, voters who prefer a party that is unlikely to reach the threshold often strategically vote for a safer party.%
\footnote{Another kind of strategic voting in proportional representation is the so-called \emph{coalition insurance voting} \citep{gschwend2007ticket,freden2017opinion,Shikano01052009}, where voters preferring a large party may strategically vote for a smaller party at risk of not meeting the threshold, in the hope that this smaller party can form a coalition with the preferred larger party. We include a brief discussion of this phenomenon in \Cref{sec:insurance}.}
Intuitively, this strategizing is less effective when voters can rank several parties. We formalize this intuition through a sequence of strategyproofness axioms. We find that STV fails those axioms in the worst case, but are able to establish that DO and GP are strategyproof in a number of politically plausible situations.

Given the incentives for voters to misrepresent their preferences in the current systems, we would expect significant changes in voting behavior if any of our ranking-based methods were to be adopted. Since there is little empirical research about strategic voter behavior in threshold elections \citep{hellmann23warum},%
\footnote{A notable exception is a study by \citet{graeb2018ersatzstimme}, which asked $n = 828$ participants in Germany in 2017 how they would vote if they could indicate a second-choice replacement vote. Like us, they find that participants vote more frequently for smaller parties, and the the number of lost votes decreases (they applied the DO rule). However, their study does not compare different party selection rules, and it introduces two different changes to the voting system simultaneously: addition of a replacement vote but also a combination of the person-bound and party-bound votes (\emph{Erststimme} and \emph{Zweitstimme}) into a single vote, making it more difficult to estimate the impact of the replacement vote alone.}
we decided to conduct our own experiment in the context of the 2024 election of the French representatives to the European parliament, motivated by the 2019 election having had a high percentage of lost votes.
Our experiment was based on a survey inviting participants to report how they would vote if they could rank several parties in the election. We collected two datasets, one with $n = 3\,046$ participants we recruited via social media, and one with $n = 1\,000$ participants recruited by a survey research company.
By comparing the ranking data to the voting intentions of our survey participants, we are able to perform a counterfactual analysis. 
In particular, we can study the strategic behavior of the voters. For example, we find that between 5\% and 29\% of voters strategically vote for a larger party in the actual election, while choosing a smaller party as their top-ranked party when given the possibility to provide a ranking.
In addition, we can simulate the behavior of the three party selection rules that we studied. The results confirm the theoretical predictions: our party selection rules would significantly reduce the number of wasted votes without fragmenting the parliament. Our results are robust to adding noise to the data, to increasing the threshold above 5\%, and to truncating the rankings to just 3 places.

\subsubsection*{Other Applications of the Model}
Although our main motivation is to design better parliamentary election systems, our formal model of party selection rules applies more generally to any multi-item selection context where each item is required to have enough voter support, where `voter support' means being the voter's preferred item within the selection. This applies whenever we have to find a clustering of voters, each cluster being associated with some item. For example, consider a university program that needs to select which optional courses to open in a particular year, given that each student will choose their preferred course from the selection, and that a course should be opened only if is taken by a minimal number of students. Other examples are selecting a set of activities to be organized for a group on a given day, such that every participant will choose one [the \emph{group activity selection problem}, \citealp{darmann2012group,DarmannDDS22}], or ordering a set of dishes for a company lunch, where everyone will eat a portion of their preferred dish out of the selection, given that the caterer will not prepare a dish in small quantities. 
Two major differences between parliamentary elections and clustering-based settings are that the latter do not involve an apportionment step, and that voters' preferences bear only on the item they are assigned to, while in parliamentary elections they usually depend on the final composition of the parliament.

\subsubsection*{Related Work}

One can conceptualize our formal model as being about multi-winner voting based on input rankings. The literature generally studies this topic with complete input rankings (not truncated) and with an output containing a fixed number of winners \citep{FaliszewskiTrends2017}, with some exceptions studying a variable number of winners \citep{FaliszewskiST20,LangMSZ17,BrandlP19}.
However, this literature has not studied thresholds.

Party-list elections are typically studied in models where voters choose a single party (which we call \emph{uninominal} voting), in particular as part of apportionment theory \citep{FairRepresentationBook,pukelsheim2014proportional}.
Some works going beyond uninominal voting are \citet{BrillGPSW24} who allow voters to approve several parties (further studied by \citet{delemazure2023strategyproofness}), and \citet[Section 7]{AiriauACKLP23} who consider apportionment based on complete rankings of the competing parties. None of these works addresses the question of thresholds.

When thinking of our model as applying more generally to item selection contexts where voters care about the item assigned to them, some other work is relevant.
In particular, our model becomes a special case of group activity selection with group size constraints studied by \citet{DarmannDDS22} who focus on stability and efficiency notions, their mutual compatibility, and the computational difficulty of finding solutions.
The model can also be seen as a type of facility location problem where facilities may only be opened if they will serve a minimum number of users. Algorithmic questions about such problems have been studied 
\citep{svitkina2010lower,ahmadian2012improved,li2019facility}. However, \emph{upper} bounds on the number of users of a facility are much more common in the literature (see \citet{AzizCLP20}).
The Monroe multi-winner voting rule, designed for proportional representation, also shares some similarities to our model, since it involves assigning voters to candidates. However, this rule operates with a fixed number of winners, and does not necessarily assign voters to their most-preferred committee member.

Strategic voting behavior in elections using proportional representation systems with or without thresholds have been observed in several empirical studies, including in a 2003 Israel election \citep{blais2006do}, in the 2010 election in Sweden involving thresholds \citep{freden2014threshold}, and in the 2014 elections in Belgium \citep{verthe2018strategic}, as well as in laboratory studies \citep{lebon2018sincere,freden2016coalitions}.

\section{Preliminaries}
\label{sec:preliminaries}

Let $C = \{c_1, \dots, c_m\}$ be a set of \emph{parties} and $N = \{1, \dots, n\}$ be a set of \emph{voters}. A \emph{truncated ranking} $\succ$ is a ranking of a subset of parties. The non-ranked parties are considered to be less preferred than the ranked ones, and to be incomparable to each other. For set of candidates $S,T \subseteq C$, we write $S \succ T$ if for all $a \in S$ and $b \in T$, we have $a \succ b$.

A preference \emph{profile} is a collection of truncated rankings $P = (\succ_1, \dots, \succ_n)$. A \emph{full} profile is a profile in which every voter ranks all the parties (so the rankings are not truncated). A \emph{uninominal} profile is a profile in which every voter ranks exactly one party.

An \emph{outcome} is a (possibly empty) subset of parties $S \subseteq C$. Together, an outcome and a profile jointly define a unique mapping $\best_S: N \to S \cup \{ \emptyset \}$ that assigns every voter $i$ to her most-preferred party $\best_S(i)$ 
among those in $S$, and to the empty set if she does not rank any party in $S$. We say that $c = \best_S(i)$ is the \emph{representative} of voter $i$ in $S$, and that $i$ is \emph{unrepresented} in $S$ if $\best_S(i) = \emptyset$.\footnote{Thus, the `representative' of a voter is a party, not an individual candidate.} For a party $c \in S$, we define the \emph{supporters} of $c$ as the set of voters of which it is the representative in~$S$, $\supp_S(c) = \{ i \in N : \best_S(i) = c \} \subseteq N$, and the \emph{score} of $c$ is $\rep_S(c) = |\supp_S(c)|$. 

Given a profile $P$ and a \emph{threshold} $\tau \in \mathbb{N}$ with $0 \le \tau \le |N|$,  an outcome $S$ is \emph{feasible} (for $P$ and $\tau$) if every party $c \in S$ has at least $\tau$ supporters, that is, $\rep_S(c) \ge \tau$ for all $c \in S$.%
\footnote{Note that in our model, the threshold is an absolute number of voters rather than a fraction. This choice makes the notation clearer. Also, we do not consider issues of abstention, and the total number of voters is fixed.}
Clearly, the empty set is feasible, and a subset of a feasible set is feasible. 
If $P$ is a full profile, then every singleton set $S = \{c\}$ is feasible, because we will then have $\rep_S(c) = |N| \ge \tau$. When $\tau = 0$, every subset of parties is feasible. When $\tau = 1$ and we have a full profile, then $S$ is feasible if and only if $S$ does not contain two parties $c$ and $c'$ such that $c$ Pareto-dominates $c'$. 
When $\tau = n$ (or more generally $\tau > \frac{n}{2}$) and we have a full profile, then $S$ is feasible if and only if it is singleton or the empty set. Hence, single-winner voting on full profiles will be a special case of our model with $\tau = n$ if we force rules to return a feasible non-empty outcome if there exists one.

Once the party selection $S$ has been determined, the parliament is made up using an apportionment method, each party $c \in S$ being represented proportionally to $\score_S(c)$. Our setup abstracts away issues arising due to apportionment: while in practice an outcome will need to be reduced to a fixed number of available parliament seats, we will not study this ``second step,'' trusting that it won't affect our conclusions in interesting ways. 

\section{Party Selection Rules}

A \emph{party selection rule} is a function $f$ that takes as input a profile $P$ and a threshold $\tau$, and returns a feasible outcome $f(P, \tau)$. In this paper, we will focus on three specific rules:
\begin{itemize}
	\item \emph{Direct winners only (DO):} This rule selects the outcome whose support consists of all parties who are ranked in top position by at least $\tau$ voters. More formally:
	\[
		f(P, \tau) = \{ c \in C : |\{ i \in N : c \succ_i C \setminus \{c\} \}| \ge \tau \}.
	\]
	\item \emph{Single transferable vote (STV):} This rule starts with the set $S_0 = C$. Then, at each step $k \ge 0$, if $S_k$ is feasible, it returns this set. Otherwise, the rule identifies the party $c \in S_k$ who is ranked first (among parties of $S_k$) by the fewest voters and sets $S_{k+1} = S_k \setminus \{c\}$.
	\item \emph{Greedy plurality (GP):} This rule starts with the empty set and goes over each party in decreasing order of plurality score (the number of voters ranking it first), adding it to the outcome set if the outcome remains feasible; otherwise it is skipped.
\end{itemize}

Except for DO, there might be ties in the rules when two parties have the same plurality scores at some point during the execution of STV or GP. In this case, we use some fixed tie-breaking order on the parties to decide how to proceed. This assumption makes it easier to analyze the rules. However, all the results we present in this paper can be extended to the case of irresolute rules with parallel-universe tiebreaking \citep{conitzer2009preference,freeman2015general}, where all possible ways to break ties are considered to obtain the set of outcomes.

DO, STV, and GP are polynomial-time computable, as well as easy to understand. Many other interesting party selection rules can be defined. For example, one can consider rules that optimize some objective function over the set of feasible outcomes, such as maximizing the number of represented voters or the number of voters whose most-preferred party is in the outcome. However, these rules are hard to compute, as well as hard to verify and understand by voters, thus not particularly suitable to be used in parliamentary elections. However, they are interesting for some other contexts covered by our model, such as group activity selection or facility location. We define and study these rules in \Cref{app:max-rules}.

On uninominal profiles, DO, STV, and GP are equivalent. However, in the general case, they can give different outcomes, as in the following example.

\begin{example} \label{ex:example_diff}
	Let $P$ be the following profile:
	\begin{align*}
		\text{4: } & a \succ b \succ c&
		\text{3: } & b \succ c &
		\text{2: } & c \succ b \succ a &
		\text{2: } & d &
		\text{4: } & d \succ b
	\end{align*}
	with the threshold $\tau = 5$. In this profile, the only party with a plurality score higher than $5$ is $d$, thus DO returns $\{d\}$. If we run the STV rule, we will eliminate $c$ first, as it has the lowest plurality score, and then $a$ because the $c$ voters are now supporting $b$, and we obtain the outcome $\{d,b\}$. Finally, if we run the GP rule, we will add $d$ first, then we will add $a$ since it is the party with the second highest plurality score and $\{d,a\}$ is feasible. However, we will not add $b$ since $\{d,a,b\}$ is not feasible ($a$ has only $4$ supporters), and same for $c$. Thus, the outcome of GP is $\{d,a\}$. 
\end{example}

The rule used in actual elections with thresholds selects exactly those parties whose plurality score is above the threshold. Therefore, as a party selection rule, it coincides with DO. However, implicitly the associated representation and score functions $\best_S$ and $\rep_S$ are different, because the uninominal system ignores all votes that did not rank one of the selected parties on top instead of assigning them to their best selected party (one can think of it as first transforming the profile by removing everything below the first party ranked by each voter, and then running DO on the resulting top-truncated profile). This difference can have a large impact on the number of wasted votes and the final composition of the parliament. We give a detailed comparison in \Cref{app:uninominal}.

\section{Axiomatic Analysis}
\label{sec:axioms}
\setlength{\textfloatsep}{0pt plus 1.0pt minus 2.0pt}
\begin{table}
	\begin{center}
		\begin{tabular}{l c c c }
			\toprule
			 & DO & STV & GP \\
			\midrule
			Set-maximal & \xmark & \xmark& \cmark\\
			Inclusion of direct winners & \cmark\rlap{*} & \cmark\rlap{*} & \cmark \\
			Representation of solid coalitions & \xmark & \cmark& \xmark \\
			Threshold monotonicity& \cmark & \cmark & \xmark\\
			Independence of definitely losing parties  & \xmark & \cmark\rlap* &  \xmark \\
			Independence of clones & \xmark & \cmark& \xmark \\
			Reinforcement for winning parties  & \cmark\rlap* & \xmark & \xmark \\
			Monotonicity  & \cmark & \xmark & \xmark  \\
			Representative-strategyproof (one risky party)  & \xmark & \xmark & \cmark  \\
			Share-strategyproof (safe first or second) & \cmark & \xmark & \xmark \\
			Share-strategyproof (representative ranked first)  & \cmark & \xmark & \cmark \\
			\bottomrule
		\end{tabular}	
	\end{center}
	\caption{Properties satisfied by the rules. The ``*'' indicates characterization results.
	\label{tab:summary}}
\end{table}
\setlength{\textfloatsep}{15pt plus 8.0pt minus 5.0pt}
In this section, we define a set of axioms that we believe are desirable for a party selection rule in the context of proportional representation with thresholds. We then analyze the different rules with respect to these axioms. We will go over different kinds of axioms. \Cref{tab:summary} summarizes the results of the axiomatic analysis.

\subsection{Efficiency Axioms}\label{sec:efficiency}

The definition of Pareto efficiency in our model is not completely clear, as this requires reasoning about voters' preference over outcomes rather than parties. We will revisit this issue in \Cref{subsec:sp}. We can define an uncontroversial axiom of \emph{weak efficiency} that forbids rules from selecting the empty outcome unless they are forced to. A stronger axiom can be defined by interpreting voters to prefer outcomes in which their representative is more preferred. In this perspective, note that if we have two outcomes $S$ and $S'$ such that $S' \supseteq S$, then $\best_{S'}(i) \succeq_i \best_S(i)$ for every $i \in V$, and if $\tau > 0$, there exists $i$ such that $\best_{S'}(i) \succ_i \best_S(i)$. This motivates the definition of the axiom of set-maximality. 

\begin{definition}
	A party selection rule satisfies \emph{set-maximality} if for every profile $P$ and threshold $\tau$, if $S = f(P, \tau)$ and $S' \supseteq S$ is feasible, then $S' = S$. It satisfies \emph{weak efficiency} if $f(P, \tau) \neq \emptyset$ whenever there exists a non-empty feasible outcome.
\end{definition}

Note that set-maximality implies weak efficiency. Note also that while set-maximality leads to more-preferred representatives, it may also lead to outcomes with more distinct parties, which as we discussed may be undesirable.
GP satisfies both efficiency axioms, but the other rules do not.

\begin{proposition}
	GP satisfies set-maximality. DO and STV fail even weak efficiency.
\end{proposition}
\begin{proof}
	For GP, let $S$ be its outcome on some profile, and assume for a contradiction that there is a party $c \not\in S$ such that $S \cup \{c\}$ is feasible. Let $S' \subseteq S$ be the set of parties that GP has added to the outcome by the time it considered $c$. Then $S' \cup \{c\}$ is feasible since $S' \cup \{c\} \subseteq S \cup \{c\}$ and feasibility is preserved under taking subsets. This contradicts that GP did not select $c$.
	For DO and STV, take the profile $P = \{ 2: b \succ c , 1:c \}$ with $\tau = 3$. Both rules return $\emptyset$, even though $\{c\}$ is feasible.
\end{proof}

\subsection{Representation Axioms}

We now turn to representation axioms, which ensure in different ways that groups of voters of size at least $\tau$ must be represented. The most basic axiom requires that all parties that receive enough first-place votes must be winners. This axiom seems essential for political applications.

\begin{definition}
	A party selection rule satisfies \emph{inclusion of direct winners} if for every profile $P$ and threshold $\tau$, whenever $c$ is a party such that at least $\tau$ voters rank $c$ in top position, then $c \in f(P, \tau)$.
\end{definition}

It is easy to see that DO, STV, and GP satisfy this axiom. Moreover, since DO returns \emph{only} the direct winners, we always have $\DO(P, \tau) \subseteq \STV(P, \tau)$ and $\DO(P, \tau) \subseteq \GP(P, \tau)$. 
Thus, there are at least as many unrepresented voters under DO as under STV or GP, by the argument in \Cref{sec:efficiency}.

One can strengthen inclusion of direct winners to apply to cases where enough voters support a \emph{set} of parties. For example, suppose that there are three ``green'' parties and that at least $\tau$ voters rank them in the top three ranks, though they may disagree on their relative ordering. The following axiom requires that at least one of the green parties is included in the outcome. It is inspired by the ``proportionality for solid coalitions'' (PSC) axiom in multi-winner voting \citep{dummett1984voting}.

\begin{definition}
	A party selection rule satisfies \emph{representation of solid coalitions} if for every profile $P$ and threshold $\tau$, if $T \subseteq C$ is a set of parties that has at least $\tau$ supporters in the sense that $|\{ i \in N : T \succ_i C \setminus T \}| \ge \tau$, then $T \cap f(P, \tau) \neq \emptyset$.
\end{definition}

This property is satisfied by STV, but it is failed by DO, since it can happen that none of the parties supported by a solid coalition has enough first-place votes. It is also failed by GP.

\begin{proposition}
	STV satisfies representation of solid coalitions, but DO and GP do not.
\end{proposition}
\begin{proof}
	For STV, assume that there is a profile $P$ where the axiom is violated, i.e. there exists a set $T \subseteq C$ of parties with $|\{ i \in N : T \succ_i C \setminus T \}| \ge \tau$, but $T \cap \STV(P,\tau) = \emptyset$. Consider the first step in the execution of STV at which all but one party of $T$ are eliminated. At that point, this party $c \in T$ is ranked first by at least $\tau$ voters such that $T \succ_i C \setminus T$. Thus, $c$ cannot be eliminated in the subsequent steps and must therefore be included in the outcome, which is a contradiction. 
	
	DO and GP fail the property on the profile $P = \{4:a\succ b \succ c, 3:b\succ c \succ a, 2: c\succ b \succ a \}$ with $\tau = 5$. The last five voters form a solid coalition for $\{b,c\}$. However, we have $\DO(P,\tau) = \emptyset$ since no party is ranked first by at least 5 voters, and $\GP(P,\tau) = \{a\}$ since $a$ is the party who is ranked first by the most voters and any set containing more than one party is not feasible, as $\tau > \frac{n}{2}$. 
\end{proof}

Note that representation of solid coalitions implies inclusion of direct winners (consider singleton $T$). One could strengthen this axiom further by forbidding that there is a party $c$ outside the outcome $S$ for which $\tau$ voters prefer $c$ to all parties in $S$. This is a version of the local stability axiom studied in multi-winner voting \citep{aziz2017condorcet,jiang2020approximatelystable}.

\begin{definition}
	A party selection rule satisfies \emph{local stability} if for every profile $P$ and threshold $\tau$, and $S = f(P,\tau)$, for all parties $c \notin S$, we have $|\{ i \in N :  c \succ_i S \}| < \tau$.
\end{definition}

Note that this axiom implies representation of solid coalitions, and thus inclusion of direct winners. However, this axiom cannot be satisfied for any  $\tau \notin \{1, n\}$.

\begin{proposition} \label{prop:local_stability}
	No party selection rule satisfies local stability for any $\tau \notin \{1,n \}$.
\end{proposition}
\begin{proof}
	For $n \ge 3$ and $\tau = n-1$, we can take a simple profile of $m=n$ parties $c_1, \dots, c_n$ with a Condorcet cycle. This corresponds to the profile $P = (\succ_1, \dots, \succ_n)$ in which $\succ_i = c_{i+1} \succ \dots \succ c_n \succ c_1 \succ \dots \succ c_{i-1}$ for all $i$.
	Because the threshold is greater than $n/2$, at most one party can be part of the outcome. However, in a Condorcet cycle every $c_i$ is better ranked than $c_{i+1}$ in $n-1$ rankings (and $c_n$ is better ranked than $c_1$ in $n-1$ rankings), thus for any feasible outcome $S = \{c_i\}$, we have $|\{ i \in V :  c_{i+1} \succ_i S \}| = n-1 \ge \tau$, and the rule fails the axiom. If $S = \emptyset$, the result is even more clear.
	
	For all other $\tau \notin \{1,n\}$, we can take a profile $P$ with $m = \tau+2$ parties $c_1,\dots,c_{\tau+1},d$, such that $P$ contains the Condorcet cycle of the example above with $m' = \tau+1$ parties and $n'=\tau+1$ voters (so with parties $c_1,\dots,c_{\tau+1}$), and add $n-n'$ dummy voters that only rank $d$. In this profile, the outcome must contain at most one of the $c_i$, as only $\tau+1 < 2\tau$ voters ranked them. However, for any $c_i \in S$, we have $|\{ i \in V :  c_{i+1} \succ_i S \}| = n'-1 = \tau$, which breaks local stability. If no $c_i$ is part of the outcome, again the result is even easier to see.
\end{proof}

Note that for $\tau = n$, the rule that selects a party with full support and which is not Pareto-dominated satisfies the axiom. Similarly, for $\tau = 1$, the rule that selects all parties that are ranked first by at least one voter satisfies the axiom.

However, we can weaken the local stability axiom by restricting the voters that can ask for a better representative to the set of unrepresented voters. In other words, no party outside the outcome should be included in the (truncated) rankings of $\tau$ \emph{unrepresented} voters. 
	Thus, a party selection rule satisfies \emph{representation of unrepresented voters} if for every profile $P$ and threshold $\tau$, there is no party $c \notin S = f(P,\tau)$ with $|\{ i \in N : \best_S(i) = \emptyset \text{ and } c\succ_i S\}| \ge \tau$.

Unfortunately, we can show that this axiom is not satisfied either by any of the three rules we consider. Still, it is possible to satisfy this axiom, and each of the three rules can be transformed into a rule satisfying it through a local search procedure. See \Cref{app:unrepresented} for details.

\subsection{Varying the Threshold}

We now discuss what should happen to the outcome when the threshold changes. The first axiom says that a losing party should stay losing if the threshold increases, and conversely a winning party should stay winning if the threshold decreases. This is a natural requirement, as a higher threshold should make it harder for a party to be selected.

\begin{definition}
	A party selection rule satisfies \emph{threshold monotonicity} if for all profiles $P$ and all thresholds $\tau \le \tau'$, we have $f(P, \tau) \supseteq f(P, \tau')$.
\end{definition}

It is easy to see that DO and STV satisfy this axiom, but GP does not. 

\begin{proposition}
	DO and STV satisfy threshold monotonicity, but GP does not.
\end{proposition}
\begin{proof}
	The result for DO is clear: if a party is ranked first by $\tau'$ voters, it is ranked first by $\tau \le \tau'$ voters. For STV, the order of elimination for the parties is the same, but the rule might stop earlier for a lower threshold; thus a subset of parties are eliminated.
	For GP, consider the profile $P = \{3: a \succ b, 2:b \}$. With $\tau = 3$, the outcome is $\{a\}$ and with $\tau' = 4$, the outcome is $\{b\}$.
\end{proof}

Proponents of STV often argue in its favor using the following principle of procedural fairness: once we have decided that some candidates are losing, we should continue the procedure as if that party hadn't run in the first place \citep{meek1969nouvelle}. We can formalize this principle using the following independence axiom, which says that once some parties are losing at some threshold, then for all larger thresholds, the rule should behave as if none of the losing parties had been available.

\begin{definition}
	A party selection rule satisfies \emph{independence of definitely losing parties} if for every profile $P$ and thresholds $\tau \le \tau'$, writing $S = f(P, \tau)$, we have $f(P, \tau') = f(P|_{S}, \tau')$.
\end{definition}

Here, $P|_S$ denotes the profile obtained from $P$ by deleting all parties outside $S$.
Independence of definitely losing parties implies threshold monotonicity, because it requires that $f(P, \tau') = f(P|_{S}, \tau') \subseteq S = f(P, \tau)$. The axiom is related to the ``independence at the bottom'' axiom of \citet{FBC14a}, and it encodes a key intuition behind the functioning of the STV rule. In fact, in combination with inclusion of direct winners, this axiom characterizes STV. To state the axiom formally, we say that a profile $P$ is \emph{generic} if for every $S \subseteq C$, in the profile $P|_S$ there is a unique party with the lowest plurality score $\score_S(c)$. On generic profiles, the STV rule never encounters a tie.

\begin{theorem}
	Let $f$ be a party selection rule satisfying inclusion of direct winners and independence of definitely losing parties. Then $f$ equals STV on all generic profiles.	
\end{theorem}
\begin{proof}
	Let $\tau \ge 0$ be a threshold and $P$ a generic profile. We use induction on the number of parties $m$. If there is just one party, then STV selects that party as winner if and only if more than $\tau$ voters ranked it, and $f$ must do the same by inclusion of direct winners and by feasibility. Assume that we have proven that $f(P, \tau) = \STV(P, \tau)$ for all generic profiles $P$ with $m - 1$ parties.

	Let $P$ be a generic profile with $m$ parties. Suppose that the party $c$ with the lowest plurality score has score $\score_C(c) = s$. For thresholds $\tau \le s$, note that inclusion of direct winners forces the outcome, and both $f$ and STV select the set of all parties $C$ as winners. Next, consider threshold $\tau = s + 1$. Because there exists a unique plurality loser, all parties other than $c$ have more than $\tau$ voters ranking them first. Thus, all parties except $c$ are direct winners and must therefore be selected by $f$, i.e. $C \setminus \{c\} \subseteq f(P, \tau)$. However, the set $C$ of all parties is not feasible because all voters who ranked at least one party will be represented by their first choice, and voters who ranked none will not be represented, thus $\score_C(c) = s < \tau$. Hence we must have $f(P, \tau) = C \setminus \{c\}$, which agrees with the STV outcome. Finally, consider a threshold $\tau > s + 1$. Then we have
	\begin{align*}
		f(P, \tau) 
		&= f(P|_{C \setminus \{c\}}, \tau) \tag{independence of definitely losing parties} \\
		&= \STV(P|_{C \setminus \{c\}}, \tau) \tag{inductive hypothesis} \\
		&= \STV(P, \tau), \tag{definition of STV and genericness}
	\end{align*}
	and thus again $f$ selects the same set as STV, as we wanted to show.
\end{proof}

Note that this characterization cannot be used as a starting point for an axiomatic characterization of STV as a single-winner voting rule, because it crucially depends on a variable-threshold setup.

\begin{corollary}
	DO and GP do not satisfy independence of definitely losing parties.
\end{corollary}

\subsection{Independence of Clones}

In addition to the independence of definitely losing parties, STV also satisfies an independence of clones axiom, as is often the case for similar rules in other models \citep{tideman1987independence}.

We say that two parties $c,c' \in C$ are \emph{clones} if for all $i \in N$ and all $x \in C \setminus \{c,c'\}$, $c \succ_i x \Leftrightarrow c' \succ_i x$ and $x \succ_i c \Leftrightarrow x \succ_i c'$. This implies that every voter ranks $c$ and $c'$ consecutively. The axiom says that when we have clones in a profile, the outcome should not change if we remove one clone. 

\begin{definition}
	A party selection rule satisfies \emph{independence of clones} if for every generic profile $P$ and threshold $\tau$, if $c$ and $c'$ are clones in $P$, writing $S = f(P, \tau)$ and $S' = f(P|_{C \setminus \{c'\}}, \tau)$ then
	\begin{enumerate}
		\item $\{c,c'\} \cap S \ne \emptyset$ if and only if $c \in S'$, and
		\item for all $x \in C \setminus \{c,c'\}$, $x \in S$ if and only if $x \in S'$.
	\end{enumerate}
\end{definition}

Note that this axiom is equivalent to the independence of clones axiom for single-winner irresolute voting rules. It is satisfied by STV, but not by the other rules.

\begin{proposition} \label{prop:clones}
	STV satisfies independence of clones, but DO and GP do not.
\end{proposition}

\begin{proof}
	Let us first show that STV satisfies independence of clones. Let $P$ be a generic profile and $P'$ another profile equivalent to $P$ but in which some party $c \in C$ has been cloned into another party $c' \notin C$. In both profiles, we will eliminate parties one by one. 

	As long as neither $c$ nor $c'$ is the plurality loser in $P'$, the order of elimination will be the same in $P$ and $P'$, since cloning $c$ does not affect the relative order of all other parties in the rankings. Note that STV stops when it finds a feasible outcome, which happens if and only if every remaining party is ranked first by at least $\tau$ voters among the remaining parties.
	
	If after some eliminations (the same in both profiles), STV finds a feasible outcome in $P'$ before eliminating either $c$ or $c'$, then this means every remaining party is ranked first among the set of remaining party by at least $\tau$ voters, and thus this is also the case in $P$ after the same elimination steps, and $\STV(P',\tau) \subseteq \STV(P,\tau)\cup\{c'\}$. Moreover, STV could not have found a feasible outcome in $P$ before this step, as the score of parties other than $c$ and $c'$ are the same in both profiles. Thus $\STV(P',\tau) = \STV(P,\tau)\cup\{c'\}$, and the conditions of the axiom are satisfied.

	Now, assume that we reach a step such that $c$ or $c'$ is the plurality loser in $P'$ before finding a feasible outcome. 
    Without loss of generality, assume that $c'$ is the plurality loser, otherwise we exchange the names of $c$ and $c'$. Then, we eliminate $c'$ in $P'$ and do nothing in $P$ (because $P$ is generic, we can assume that $c'$ is eliminated in $P'$). We now reached a step such that the two profiles are perfectly equivalent as the set of remaining parties are identical, and thus after this point the outcomes of STV on $P$ and $P'$ will necessarily be the same (remember that $P$ is generic, so we will never encounter any tie). This concludes the proof that STV satisfies independence of clones. 

	To show that DO and GP do not satisfy independence of clones, consider the following generic profile $P = \{6:a, 4:c'\succ c \succ a, 3: c\succ c' \succ a\}$ with $\tau = 7$. In this profile, $\DO(P,\tau) = \emptyset$ and $\GP(P, \tau)  =  \{a\}$. However, $c$ and $c'$ are clones. Denote $P' = \{6:a,7:c \succ a\}$ the profile without the clone $c'$. In this profile, $\DO(P',\tau) = \GP(P',\tau) =  \{c\}$, contradicting independence of clones. Note that for GP, we could also deduce it from the fact that with $\tau = n$ and a full profile, it is equivalent to the plurality rule in the single-winner setting, which is not independent of clones.
\end{proof}

When using parallel-universe tie-breaking, STV satisfies independence of clones even without restricting the axiom to generic profiles. %

\subsection{Reinforcement for Winning Parties}

The next axiom connects the outcomes of a party selection rule on profiles defined on different sets of voters, and imposes a consistency condition. The axiom is inspired by the reinforcement axiom introduced by \citet{young1974axiomatization}. Our axiom says that if a party is winning in a profile $P_1$ with a threshold $\tau_1$ and in a profile $P_2$ with a threshold $\tau_2$, then it should also be winning in the profile $P_1 + P_2$ with threshold $\tau_1 + \tau_2$, where $P_1 + P_2$ is the profile obtained by ``concatenating'' $P_1$ and $P_2$.

\begin{definition}
	A party selection rule satisfies \emph{reinforcement for winning parties} if for all profiles $P_1$ and $P_2$ and all thresholds $\tau_1$ and $\tau_2$, if $c \in f(P_1, \tau_1)$ and $c \in f(P_2, \tau_2)$, then $c \in f(P_1 + P_2, \tau_1 + \tau_2)$.
\end{definition}

DO satisfies this. In fact, DO is the only party selection rule satisfying inclusion of direct winners and reinforcement for winning parties, thus providing an axiomatic characterization of this rule.%

\begin{theorem}
	DO is the only party selection rule that satisfies inclusion of direct winners and reinforcement for winning parties.
\end{theorem}
\begin{proof}
Suppose for a contradiction that there is a party selection rule $f$ other than DO that satisfies these axioms. Since it differs from DO, there is some profile $P_1$ and threshold $\tau_1$ such that party $c$ is not an direct winner but $c \in f(P_1, \tau_1)$. Now build the profile $P_2$ where each party except $c$ gets $\tau_1 + 1$ voters ranking only that party, and $c$ gets 1 voter ranking only $c$. Set $\tau_2 = 1$. Then every party is an direct winner in $P_2$ and thus by inclusion of direct winners, $f(P_2, \tau_2) = C \ni c$. Hence by reinforcement for winning parties, $c \in f(P_1 + P_2, \tau_1 + \tau_2)$. Note that in $P_1 + P_2$, the party $c$ is not an direct winner (because in $P_1$ it had strictly fewer supporters than $\tau_1$, while in $P_1 + P_2$ it has one more supporter than before, which is strictly less than $\tau_1 + \tau_2 = \tau_1 + 1$). However, all other parties are direct winners since they have at least $\tau_1 + 1$ direct votes in $P_2$. Thus by inclusion of direct winners, $C \setminus \{c\} \subseteq f(P_1 + P_2, \tau_1 + \tau_2)$. But note that in $P_1 + P_2$, the set of all parties $C$ is not feasible, because every voter who is not a supporter of $c$ must be assigned to their favorite party, so $c$ can only be assigned voters who rank $c$ top, of whom there exist strictly fewer than $\tau_1 + \tau_2$ voters. Thus, $c \not\in  f(P_1 + P_2, \tau_1 + \tau_2)$, a contradiction.
\end{proof}

\begin{corollary}
	STV and GP do not satisfy reinforcement for winning parties.
\end{corollary}

\subsection{Monotonicity}

The next axiom is the classical monotonicity axiom which says that if a party is selected in the outcome, then it should remain selected if some voters place this party in a better position in their vote, leaving unchanged the relative ranking of the other parties. 

\begin{definition}
	A party selection rule satisfies \emph{monotonicity} if for every profile $P$ and threshold $\tau$, if $c \in f(P, \tau)$ and $P'$ is obtained from $P$ by increasing the rank of $c$ in one ranking without changing the relative ordering of the other parties, then $c \in f(P', \tau)$.
\end{definition}

Note that this axiom corresponds to monotonicity for single-winner irresolute voting rules, which is known to be failed by the single-winner version of STV. 

\begin{proposition}
	DO satisfies monotonicity, but STV and GP do not.
\end{proposition}
\begin{proof}
	For DO, consider a profile $P$ and a threshold $\tau$, and suppose that $c \in \DO(P, \tau)$. Now, consider a profile $P'$ obtained from $P$ by increasing the rank of $c$ in one ranking without changing the relative ranking of the other parties. Then, the voters who ranked $c$ in top position in $P$ still rank $c$ in top position in $P'$, and thus $c$ still has at least $\tau$ supporters in $P'$, so $c \in \DO(P', \tau)$.
	
	For STV, take the profile $P = \{5:a \succ c, 6: c, 13:d, 4:b \succ a, 2: b\succ c \}$, and $\tau = 13$.
	In $P$, $a$ is eliminated first, then $b$, and the outcome is $\{c,d\}$. Now, let $P' = \{5:a\succ c, 6:c, 13: d, 4: b \succ a, 2: \underline{c} \succ \underline{b}\}$ be the profile obtained from $P$ by increasing the rank of $c$ in the last two rankings. 
	In $P'$, $b$ is now eliminated first, then $c$, then $a$ and the outcome is $\{d\}$. Thus, $c \notin \STV(P', \tau)$, contradicting monotonicity.

	For GP, let $P = \{5: a \succ c, 2: a \succ b \succ c, 6: c \succ b, 2:b\}$, with $\tau =7$. $a$ is added to the committee first. Then $c$ is considered, but $\{a,c\}$ is not feasible. Then $b$ is added to the committee, as $\{a,b\}$ is feasible. Now, let $P' = \{5: a \succ c , 2: \underline{b} \succ \underline{a} \succ c, 6: c \succ b, 2 :b \}$ obtained by increasing the rank of $b$. In $P'$, $c$ is added first to the committee, then $a$ is considered, but $\{c,a\}$ is not feasible. Then $b$ is considered, but $\{c,b\}$ is not feasible. Thus, $b \notin \GP(P', \tau)$, contradicting monotonicity.
\end{proof}

\subsection{Incentive Issues}\label{subsec:sp}

A major drawback of uninominal voting is that it incentivizes voters to strategically misreport their preferences. 
In particular, in standard uninominal elections, voters whose favorite party will not reach the threshold may instead vote for a large party (and thereby increase its share of the parliament) or vote for a party near the threshold (and thereby potentially move it above the threshold). We refer to this type of strategic voting as \emph{tactical voting}. We will study the extent to which tactical voting can be avoided when using party selection rules.

A second distinct type of manipulation has been called \emph{coalition insurance voting} \citep{gschwend2007ticket,freden2017opinion,Shikano01052009}. Here, a voter whose preferred party is guaranteed to reach the threshold decides to instead vote for a less-preferred party that is in danger of missing the threshold. If that smaller party reaches the threshold, it may form a governing coalition with the voters preferred party. Thus, while the voter is now contributing their support to a worse representative, the voter will be more satisfied with the parliament as a whole. This manipulation is specific to parliamentary elections with thresholds and has no analogue in single-winner voting. We will consider this second type of manipulation separately.

\subsubsection{Tactical Voting} \label{subsec:tactical}

Consider the common tactical vote in a uninominal election: 
a voter that supports a sub-threshold party instead votes for a larger party in order to  have her vote counted towards the parliament's final composition.
By doing this, the voter 
increases the share of representation of the party she votes for without any cost to her true favorite party (which anyway was not going to meet the threshold). In the extreme case, a tactical vote can even cause a new party to enter the winning set.
In our model, we can define tactical voting in two natural ways, depending on how we measure the satisfaction of the voter with an outcome $S$:%
\begin{enumerate}
	\item The satisfaction corresponds \emph{solely} to the highest position of a party in $S$ in the voter's truthful ranking.\footnote{We note that this notion is particularly natural in a clustering-based context, where an agent prefers an outcome to another whenever they are assigned to a better representative.}
	\item The satisfaction corresponds to the highest position of a party in $S$ in the voter's truthful ranking, \emph{and} to the share of representation $\share_S(c)$ of that party, where %
	the share of representation of a party $c$ is defined as $\share_S(c) = \rep_S(c)/(\sum_{x \in S}\rep_S(x))$.
\end{enumerate}

Of course, other notions of satisfaction are possible and could in principle depend on the entire vector of shares of representation and the voter's truthful truncated ranking. However, the two notions above will be sufficient to capture typical manipulations while being permissive enough to allow for positive results.
Note that obtaining full strategyproofness with respect to either of these notions is hopeless. Assume that $\tau=n$ and that all voters submit complete rankings. Then this is almost exactly the single-winner case (since sets with more than one party are not feasible), and we know from the Gibbard--Satterthwaite theorem~\citep{gibbard1973manipulation,satterthwaite1975strategy} that any rule that is strategyproof is either a dictatorship (i.e., the outcome is always the favorite party of some voter $i \in N$) or imposing (i.e., some alternative is never selected). We prove a version of this result formally in \Cref{app:strategyproofness}. 

We will therefore consider strategyproofness in restricted cases, considering a sequence of politically plausible situations where some of our rules turn out to be immune to manipulations.
The notion of strategyproofness that we will first explore in such restricted cases is the following. It requires that no voter can improve her most-preferred party among those selected.

\begin{definition}
	For any profile $P$, any threshold $\tau$, any voter $i$, and any misreport $\succ_i'$, let $P'=(\succ_1, \ldots, \succ_i', \ldots, \succ_n)$, let $S=f(P,\tau)$, and let $S'=f(P',\tau)$. A party selection rule is \emph{representative-strategyproof} if $\best_{S'}(i) \not{\succ_i} \best_{S}(i)$.
\end{definition}

In search of positive results, we will make statements that distinguish three types of parties from the perspective of a voter $i$ on a particular profile, assuming that all other votes are fixed.

\begin{itemize}
	\item A party is \emph{safe} %
	if it is included in the outcome no matter how $i$ votes. %
	\item A party is \emph{risky} %
	if it might be included or not in the outcome depending on how $i$ votes. %
	\item A party is \emph{out} %
	if it is not included in the outcome no matter how $i$ votes. %
\end{itemize}  

Note that this partition of the parties depends on the rule that is used. However, for the three rules we consider, parties that are ranked first by more than $\tau+1$ voters are always safe. On the other hand, parties that are ranked by fewer than $\tau-1$ voters among the voters who did not rank a safe party first are always out. The remaining parties can be either safe, risky or out. Intuitively, risky parties are the ones that are neither clear winners nor clear losers. In a real-world scenario, this corresponds to the parties that are close to the threshold according to the polls.%

We begin by considering the case where there exists at most one risky party. This is often not an unrealistic assumption.\footnote{For example, in the 2023 New Zealand general election, with a threshold of 5\%, only one party's vote share fell in the interval $(3.08\%,8.64\%)$. The same is true in 2020 for the interval $(2.60\%,7.86\%)$ and in 2017 for $(2.44\%,7.20\%)$.}
However, uninominal voting fails representative-strategyproofness even in this case since a voter who prefers a party that is out can instead cast their vote for a risky party that is their second choice, causing the risky party to be included.

Interestingly, we can show that set-maximality implies representative-strategyproofness when there is at most one risky party.

\begin{proposition}
	\label{prop:threshold:maximality-strategyproofness}
	 If a party selection rule satisfies set-maximality, then it satisfies representative-strategyproofness when there is at most one risky party from the perspective of each voter.
\end{proposition}
\begin{proof}
Let $f$ be a party selection rule that satisfies set-maximality, $P$ a profile in which there is one risky party from the prespective of every voter, and $i \in V$ a voter that successfully manipulate by misreporting, giving another profile $P'$. Let $S = f(P,\tau)$ and $S' = f(P', \tau)$. We assume $\best_{S'}(i) \succ_i \best_{S}(i)$, and thus $\best_{S'}(i) \notin S$. Thus, $c' = \best_{S'}(i)$ is a risky party from the perspective of $i$. However, since $c'$ is better ranked than $\best_{S}(i)$ in the ranking of voter $\succ_i$ and $S'$ is feasible in $P'$, it is also feasible in $P$. Thus, by set-maximality, $S' \ne S \cup \{c'\}$, but this means that $S \setminus S' \ne \emptyset$, and thus there exists a risky party in $S$ different than $c'$, a contradiction.
\end{proof}

A direct corollary is that GP, MaxP and MaxR satisfy representative-strategyproofness when there is at most one risky party from the perspective of each voter. We can show that it is not the case of DO and STV.

\begin{proposition} \label{prop:representative-strategyproofness}
GP satisfies representative-strategyproofness when there is at most one risky party from the perspective of each voter. DO and STV do not.
\end{proposition}
\begin{proof}
	For GP, this is a direct consequence of \Cref{prop:threshold:maximality-strategyproofness}. 
	For STV, consider the profile $P = \{1: b \succ a, 1: c \succ a\}$ with $\tau=2$. Parties $b$ and $c$ are out, while $a$ is risky. We have STV$(P,\tau) = \emptyset$, but either voter could manipulate the outcome by placing $a$ first in their vote, which would result in the outcome $\{ a \}$ (assuming that ties are broken in favor of $a$).
	For DO, consider the profile $P = \{1: b \succ a, 1: a \succ c\}$ with $\tau=2$. We have DO$(P,\tau) = \emptyset$, but the first voter can manipulate by voting $a \succ b$, resulting in outcome $\{ a \}$.
\end{proof}

Note that GP fails strategyproofness if we allow two risky parties. To see this, consider the profile $P = \{3: a, 2: b \succ c, 1: c \succ b \succ a\}$ with $\tau=3$. Party $a$ is safe but $b$ and $c$ are risky from the perspective of the last voter. GP selects $\{ a, b \}$ and is then unable to add $c$. However, if the last voter reports $c \succ a \succ b$, GP will first add $a$, then consider $b$ but not add it (since $\{ a, b \}$ is no longer feasible), and finally add $c$ and output $\{ a, c \}$, which is better for the manipulating voter.

We now turn to the second, stronger notion of strategyproofness, where the satisfaction of a voter depends on both the position and the share of representation of their favorite selected party. Recall that the share of representation of a party $c$ is defined as $\share_S(c) = \rep_S(c)/(\sum_{x \in S} \rep_S(x))$.

\begin{definition}
For any profile $P$, any threshold $\tau$, any voter $i$, and any misreport $\succ_i'$, let $P'=(\succ_1, \ldots, \succ_i', \ldots, \succ_n)$, let $S=f(P,\tau)$, let $S'=f(P',\tau)$, and let $c=\best_S(i)$. A party selection rule is \emph{share-strategyproof} if $\best_{S'}(i) \not \succ_i c$ and $\share_{S'}(c) \le \share_S(c)$.
\end{definition}

First, we consider the case where every voter has a safe party among their top two preferences. Clearly, uninominal voting fails share-strategyproofness in this case, as a voter can elevate a safe party that is their second choice ahead of a party that is out but their true favorite, thus increasing the share of the vote allocated to the safe party. Notably, this strategy is ineffective under DO.

\begin{proposition}
\label{thm:large-first-or-second}
DO satisfies share-strategyproofness whenever the most-preferred or second-most-preferred party of every voter is safe from the perspective of that voter. GP and STV do not.
\end{proposition}

\begin{proof}
	For DO, consider a voter $i$ with preference $c_j \succ_i c_k \succ_i \dots$.  If $c_j$ is selected then $i$ has no incentive to misreport, as it is clear that no manipulation will increase the share of representation of $c_j$ (in fact, they will all decrease it). If $c_j$ is not included, then $c_k$ is safe by assumption, and therefore selected. Moreover, it is the representative of voter $i$. Since any manipulation by $i$ can only decrease the number of first-place votes received by $c_j$, it is impossible for $c_j$ to be selected after the manipulation: the party is \emph{out}. Similarly, any manipulation by $i$ can only decrease the number of supporters of $c_k$ without decreasing the support of any other parties.
	Therefore, no manipulation can result in an increase of share of representation for $c_k$.

	For GP, we consider the profile $P = \{ 1: a \succ b, 5: a, 3: b \succ c, 3:c \succ a \}$ with $\tau = 4$. GP selects $S=\{ a,c \}$ and we have $\share_S(a) = 6/12$. Note that party $a$ is safe from the perspective of every voter since it is a direct winner even if it loses a vote. Furthermore, party $c$ is safe from the perspective of the voters with preference $b \succ c$, since no unilateral deviation from any of them can cause $b$ to be selected, and therefore $c$ will be feasible to add to the set $\{ a \}$ with at least five supporters (even if the deviator does not rank $c$). Therefore, every voter has a safe party in first or second position and this profile satisfies the conditions of the proposition. However, if the first voter changes their report to $b$ then GP selects $S'=\{ a, b \}$ and we have $\share_{S'}(a) = 8/12>6/12$. 
	
	For STV, the same example as for GP also demonstrates a violation of share-strategyproofness, provided that ties between $b$ and $c$ are broken in favor of $b$ being eliminated.
\end{proof}

Finally, we prove a result in which the set of possible misreports of a voter is restricted. In particular, we assume that voters will only misreport by promoting their most-preferred party that is selected under truthful voting into first position in their misreport. This restriction can be thought of as giving voters perfect knowledge of which parties are out (and therefore not worth voting for), but not enough sophistication to perform arbitrarily ``complex'' manipulations.

\begin{proposition}
\label{thm:share-sp-misreport-restriction}
	DO and GP are share-strategyproof under the restriction that $c = \best_S(i)$ is ranked first in $\succ_i'$, but STV is not. %
\end{proposition}

\begin{proof}
For DO, it is clear that if a voter elevates $\best_S(i)$ to first position then neither the set of direct winners nor their supporters undergo any change.

For GP, any misreport allowed by the theorem results in $c$ receiving one additional plurality vote, while $i$'s true first choice $c_j$ receives one fewer. GP may now consider $c$ earlier and $c_j$ later than under truthful voting, but the relative order of consideration for all other parties remains unchanged. It is easy to see that $c \in S'$ and $c_j \not \in S'$. For any other party $c_k$, assume by induction that all parties considered before $c_k$ are selected in $S'$ if and only if they are selected in $S$. If $c_k \in S$ then it must be feasible to add $c_k$ to $S'$, since the outcome set at the time that $c_k$ is considered is a strict subset of $S$.  Similarly, if $c_k \not \in S$ then it must not be feasible to add $c_k$ to $S'$, since the outcome set at the time that $c_k$ is considered in the misreported instance is a superset of the outcome set at the time that $c_k$ is considered in the truthful instance.

	For STV, consider the profile $P = \{10:b, 4: c\succ d, 3: d\succ c, 2: d\succ b, 3: a \succ b \succ d, 1: a \succ b \succ c \succ d\}$ with $\tau = 10$.
	In this profile, %
	$c$ is eliminated first, then $a$, then the outcome $S = \{b,d\}$ is feasible with $\rep(b) = 11$ and $\rep(d) = 12$. Now, if the last voter changes her vote to $b \succ a \succ c \succ d$, then $a$ is eliminated first, then $d$ and the outcome is $S' = \{b,c\}$ with $\rep(b) = 13$ and $\rep(c) = 10$. Thus, the voter is more satisfied because $\share_{S'}(b) = 13/23 > 11/23 = \share_S(b)$.
\end{proof}

It should be observed that \Cref{thm:share-sp-misreport-restriction}
covers the typical manipulation that occurs with uninominal voting: voters put their favorite \emph{safe} party first instead of their favorite party, in order to increase the vote count of this party (otherwise, their ballot would not be considered). Consider for instance the profile $P = \{1:a \succ b, 3:b,3:c\}$ with $\tau = 3$. With our rules, the first voter can vote sincerely: her vote will support $b$ since $a$ is out. On the other hand, under uninominal voting, she has an incentive to vote for $b$. 

Finally, note that STV is not covered by any of these positive results. However, in cases where STV is manipulable, 
the voter causes the elimination order to change, and thus some parties not to be added to the outcome anymore, possibly increasing the vote share of her representative. This is arguably a very unnatural manipulation for voters, who need almost full knowledge of the preferences of the other voters to be able to predict the correct manipulation.

\subsubsection{Coalition insurance voting}
\label{sec:insurance}

A separate type of manipulation is \emph{coalition insurance voting} \citep{cox1997making}, where a supporter of a safe party $c$ instead votes for a risky party $d$ that she also likes in order to push $d$ over the threshold, thereby potentially allowing $c$ and $d$ to form a governing coalition (while in a parliament without $d$, there would be no majority for $c$).
Indeed, while $c$ loses one supporter, party $d$ gains $\tau$ supporters by virtue of being included in the outcome. In many cases, the voter will be more satisfied with the new outcome as a whole.
In many countries with proportional representation systems, parties announce intended coalitions in advance of the elections, and safe parties such as $c$ might even encourage their supporters to vote for $d$ instead.
coalition insurance voting has been observed in several countries including Germany and Sweden, and is well-studied using survey and lab experiments \citep[e.g.,][]{freden2017opinion,freden2024insurance}.

Can such voting behavior be avoided when using party selection rules?
Unfortunately, a simple example suffices to show that no rule that satisfies the inclusion of direct winners axiom is immune to manipulations of this type. Consider the profile $P = \{3:a, 3:b, 4:c, 	2:d, 1: c\succ d\}$ with $\tau = 3$, and suppose that the last voter likes both parties $c$ and $d$ and is close to indifferent between them, but dislikes $a$ and $b$.
The outcome under truthful voting is $S=\{a,b,c\}$ by inclusion of direct winners as all these parties have at least $\tau$ first-place votes, and including $d$ would violate feasibility. 
Now, the last voter only likes party $c$ from $S$, which makes up $5/11 \approx 45\%$ of the parliament, and $\{a,b\}$ may form the governing coalition.
Now, if she changes her vote to $d \succ c$, then the outcome will be $S'=\{a,b,c,d\}$ by inclusion of direct winners as all parties get at least $\tau$ first choices. After this manipulation, the last voter likes $7/13 \approx 54\%$ of the parliament, and $\{c,d\}$ may now be forming a governing coalition, thus including the most-preferred party of the manipulating voter. %

While party selection rules cannot completely avoid this effect, one might expect that in practice there is less motivation for this kind of manipulation in the case that voters can submit a ranking than if they can only submit a uninominal ballot. In particular, if the smaller party $d$ is not selected, 
then presumably many votes cast for $d$ will transfer to $c$ as a second or third choice.
For instance, in the previous example, the $d$ voters might have put $c$ in second place if $c$ and $d$ were running on similar platforms or had announced an intention of forming a coalition.
However, there is also a possibility that coalition insurance voting might \emph{increase} under rankings, since this strategy is less risky in this situation. Indeed, if voters preferring $c$ instead cast the ranking $d \succ c$, then either the manipulation is successful (and $d$ is selected), or it is unsuccessful and the vote is transferred to $c$, which does not hurt the manipulator. In contrast, under uninominal voting, a vote for $d$ carries a risk that the vote will be lost. Further experiments are needed to determine the actual impact of ranking ballots on this kind of manipulation.

\section{Experiments}
\label{sec:experiments}

In the classical uninominal system, voters either vote sincerely, at the risk of wasting their vote, or they choose to vote strategically. 
Inspired by our theoretical analysis, our goal is to check empirically if under a ranking-based systems, voters would vote less strategically, for instance by voting for less popular parties. We also want evaluate the extent to which these systems allow better representation by decreasing the number of unrepresented voters, without drastically increasing the number of parties included in the parliament (as the main argument for the threshold is to reduce the number of parties).

We base our study on a survey we ran in the context of the election of the French representatives to the European Parliament held in June 2024. The French representatives are elected by a nationwide party-list proportional representation system, using the D'Hondt apportionment method with a 5\% threshold. 
In the 2024 election, 38 lists took part in the election and 12.08\% of votes were cast on lists that did not reach the threshold. Seven lists reached it, two of which were just above the threshold (with 5.47\% and 5.5\% of vote share), so the proportion of wasted votes could have been much worse (which happened in 2019, when it was around 20\%).

\subsection{The Datasets}

Our datasets \citep{zenodo} were collected through a survey that invited participants to consider how they would vote in a system that allowed them either to rank at most two lists, or to rank an arbitrary number of lists. The survey was administered through a custom-built online platform in French language.
Participants were led through several steps (shown through screenshots in \Cref{app:screenshots}):
\begin{enumerate}
	\item Participants were briefly informed about the problem of ``wasted votes'', with the help of data visualizations of the results of the 2019 election where 19.8\% of votes had been lost due to the threshold. We then familiarized participants with the 38 lists participating in this election: for each list, we displayed its official posters and provided access to their campaigns manifestos. We sampled a random ordering of the lists for each participant, and we always displayed the lists in that order in the subsequent steps. 
	\item In the second step, participants were told that to avoid losing votes for lists that don't reach the threshold, they were allowed to specify a second choice if they wished so, which would be taken into account if their first-choice fell below the threshold of 5\%. We then asked them to indicate how they would vote under this system. Participants could rank 0, 1, or 2 lists. Paid participants had to rank at least 1 list.
	\item In the third step, we explained further that voters could now vote for as many lists as desired. We explained that if the first choice were to receive less than 5\% of the votes, the second-choice vote would be counted. If the second choice still receives less than 5\%, the third-choice vote would be counted, and so on. We then asked participants to indicate how they would vote under this system. Participants could rank between 0 and 38 lists. Paid participants had to rank at least 1 list.
	\item In the next step, participants were asked to say for which list they intended to vote on election day (for participants taking the survey prior to election day), or for which list they had voted (for those taking it after election day). Answering was optional.
	 \item Finally, participants were shown two questionnaires: one in which they could express their opinion about moving to more expressive voting, and a second one in which they could give some socio-demographic data.
\end{enumerate}

We ran this survey on two different samples of participants: 
\begin{itemize}
	\item \emph{Self-selected sample.} We recruited 3\,046 participants through social media and mailing lists. These participants tended to be interested in the political process and were very diligent in answering the survey. However, the resulting 
	sample is clearly not representative of the French population: it is very left-leaning (66\% intended to vote for one of the four biggest left-oriented parties LFI, PCF, PS and EELV, while they receive less than 32\% of the votes at the election), young (72\% are between 18 and 39, while 30\% of the French population is between 20 and 39), and highly educated (92\% indicated a university education level, while it is of 42\% in the general population). %
	Participants were not paid. This dataset only includes answers from those who indicated being registered to vote. It was gathered between May 30th and June 26th 2024, however for our analysis we will only keep the 2\,840 participants who completed the survey before election day (June 6th).
	\item \emph{Representative sample.} We recruited 1\,000 participants through the survey research company Dynata. Participants were paid a fixed amount of money to participate in the study. This sample is more representative of the French population with respect to demographics and voting behavior (56\% indicated a university education level, 33\% are between 18 and 39, and 25\% voted for one of the main four left-oriented parties). However, it is of lower quality as some participants appear to have filled out the form as quickly as possible. %
	This dataset was collected between June 17th and 25th, around two weeks after the election.
\end{itemize}

All the collected data were anonymized. The participants were informed before participating that the data would be used for research purposes only. Participants had the possibility to skip any question they did not wish to answer, except that paid participants in the representative sample had to answer the demographic questions, and rank at least one list in the second and third steps. The study received ethics approval (Université Paris Dauphine~-~PSL, décision du comité d’éthique de la recherche n°2024-01 and University of Virginia IRB protocol number 6756).
The dataset is available at \url{https://zenodo.org/records/13828295} \citep{zenodo}.

For both samples, we used the voting intentions (or actual votes if the election took place before the experiment) to assign weights to voters in order to reduce (but of course not eliminate) the representation bias of our samples. 
This weighting methodology is standard for surveys of alternative voting methods \citep{vanDerStraeten2013voteaupluriel,graeb2018ersatzstimme}.
Participants who did not provide any voting intention were assigned weight zero, leaving $n = 2\,838$ participants with non-zero weight for the self-selected sample and $n = 895$ for the representative sample. 
This gives us a total of four preference profiles: self-selected and representative samples, together with either two votes or with rankings.

\subsection{Analysis of the Results}

We conducted several analyses and experiments using these datasets. %
We first discuss the strategic behavior of the participants, then we present the results of our rules on the datasets. 

\subsubsection{Strategic Voting}
\setlength{\textfloatsep}{0pt plus 1.0pt minus 2.0pt}
\begin{table}[!t]
	\scalebox{0.9}{
	\begin{tabular}{l c c  c c} \toprule 
		& \multicolumn{2}{c}{Self-selected} & \multicolumn{2}{c}{Representative} \\
		& 2 votes & Ranking & 2 votes & Ranking \\\midrule
		Inconsistent ($c^*$ not ranked) & $3.6\%$ & $0.7\%$ &  $8.6\%$& $10.4\%$ \\ \midrule
		Sincere ($c = c^*$) & $73.5\%$  & $70.2\%$ & $84.3\%$ & $80.4\%$ \\\midrule
		Strategic ($\score(c) < \score(c^*)$) &  $22.6\%$ &$28.6\%$ & $5.0\%$ & $6.2\%$ \\\midrule
		$\quad\hookrightarrow$ out $\rightarrow$ safe & $16.5\%$ & $21.0\%$ & $2.7\%$ &$3.2\%$ \\
		$\quad\hookrightarrow$ out $\rightarrow$ risky & $2.1\%$ & $2.9\%$ & $0.3\%$ & $0.7\%$\\
		$\quad\hookrightarrow$ risky $\rightarrow$ safe & $2.6\%$ & $3.1\%$ &$0.8\%$ & $0.4\%$ \\
		$\quad\hookrightarrow$ others & $1.5\%$ & $1.5\%$ & $1.3\%$  & $2.0\%$ \\\midrule
		Strategic ($\score(c) > \score(c^*)$) & $0.3\%$ & $0.6\%$ & $2.1\%$ & $3.0\%$ \\%
		\bottomrule
	\end{tabular}}
	\smallskip
	\caption{Comparisons of the voting intention and the first ranked list in the ranking. }
	\label{tab:comparison_vote_favorite}
\end{table}
\setlength{\textfloatsep}{15pt plus 8.0pt minus 5.0pt}

To examine strategic behavior of voters in the actual election,
we compare the party $c$ they put in first position of their ranking to the party $c^*$ they actually voted for, or intended to vote for at the election.
Here, we are making the assumption that $c$ is the true top choice: even though the ranking-based rules are still manipulable, the voters have no real incentive to vote strategically, as there are no stakes to the survey. Moreover, participants were not familiar with the ranking-based rules, and would have difficulty knowing how to strategize. Thus, we assume a voter voted sincerely in the actual election if $c$ and $c^*$ are equal, and strategically otherwise.

For this analysis, we divide the parties into three groups according to their vote share in the actual election. Intuitively, we want this partition to reflect the safe/risky/out categorization that we introduced in \Cref{subsec:tactical}. We have
\begin{itemize}
	\item  5 \emph{safe} parties that received at least 7\% of the votes at the actual election,
	\item 2 \emph{risky} parties that received between 5\% and 6\% of the votes and were thus in danger of not reaching the threshold, and
	\item  31 \emph{out} parties that received less than 3\% of the votes.
\end{itemize}
This classification based on the vote shares matches expectations from polls: the least-voted safe party (LR) had been assigned at least 6\% (and usually 7--7.5\%) in all polls in the months preceding the election, while the two risky parties (EELV and REC) polled at 5--6\% \citep{wikiFranceElection}.
Then we divided the voters into four categories depending on their voting intention at the election $c^*$ and their presumed favorite party $c$. The percentage of voters in each category for the different datasets is given in \Cref{tab:comparison_vote_favorite}.

\begin{enumerate}
        \item The \emph{inconsistent} voters are those who did not include $c^*$ in their ranking. We call them inconsistent %
		as it seems irrational to choose a list when you can vote for only one, but not when you can vote for several, especially when the number of ranks is unlimited.\footnote{For the participants who completed the survey \emph{after} the election, they might have changed their mind about the list they voted for, but we assume that this is a minor effect.
		}
We can see in \Cref{tab:comparison_vote_favorite} that there are many fewer inconsistent voters in the self-selected sample than in the representative one. %
        \item The presumably \emph{sincere} voters who ranked $c^*$ in first position (i.e., $c^* = c$), and who correspond to the large majority of participants.
        \item The \emph{strategic} voters who rank in first position a party $c$ with a lower score (that is, vote count) than $c^*$. 
For a majority of them, 
$c$ is an out party while $c^*$ is a larger party (either risky or safe); this is the canonical example of ``tactical voting'' (cf. \Cref{subsec:sp}).
        \item The other \emph{strategic} voters who rank in first position a party $c$ with a higher score than $c^*$. This behaviour is less common. For a majority of these voters, $c$ is safe and $c^*$ is either risky 
        (which resembles patterns observed in coalition insurance voting, see \Cref{sec:insurance})
        or also safe. Coalition insurance voting may be less common in France compared to other countries which use proportional representation systems more frequently.
\end{enumerate}

The fact that voters are more likely to vote for an out party when they can cast rankings than when they have to vote for only one party can be shown to be statistically significant (for the representative sample, we have  $\chi^2(1,N=895) = 9.8, p = 0.002$ for length-two rankings and $\chi^2(1,N=895) = 16.2, p < 0.001$ for unlimited-length rankings).

From the questionnaire at the end of the survey we also know that the question ``Would you be more likely to vote for a small list if you could give additional choices?'' has been answered positively by 75$\%$  and 52$\%$ of the participants respectively in the self-selected and representative sample (and negatively by 19$\%$ and 29$\%$), while ``Would you be likely to vote for a small list closer to your interests even if there is a chance for your vote to be not taken into account?'' was answered negatively by 60$\%$  and 37$\%$ of participants respectively  in the self-selected and the representative sample (and positively by 33$\%$ and 47$\%$). These answers confirm that voters are indeed strategizing in uninominal elections and that the possibility to rank more parties limits the need to do so.

\subsubsection{Representativity}

We now compare the results of our different rules with that of the actual election, in particular the representativity of the outcomes.%

\paragraph{Lost votes}

\begin{figure}[t]
	\centering
	\begin{subfigure}{0.49\textwidth}
		\scalebox{0.85}{
		\begin{tabular}{r c c c c} \toprule
			& \multicolumn{2}{c}{Self-selected} & \multicolumn{2}{c}{Representative} \\
			& 2 votes & Ranking & 2 votes & Ranking \\\midrule 
		Actual & 12.1$\%$ &  12.1$\%$& 12.1$\%$ & 12.1$\%$ \\\midrule 
		Uninominal & 37.9$\%$& 34.4$\%$ & 20.5$\%$ & 21.4$\%$ \\\midrule 
		DO & 11.7$\%$ & 3.2$\%$ & 10.8$\%$& 8.7$\%$\\
		STV & 7.0$\%$& 2.3$\%$& 9.2$\%$& 7.2$\%$\\
		GP &7.0$\%$& 2.3$\%$& 9.2$\%$& 7.2$\%$\\\bottomrule
	\end{tabular}}
	\vspace{5pt}
	\caption{Percentage of unrepresented voters.}
	\label{tab:results_unrepresented}
	\end{subfigure}
	\hfill
	\begin{subfigure}{0.47\textwidth}
		\centering
		\scalebox{0.85}{
			\begin{tabular}{r c c c c} \toprule
				& \multicolumn{2}{c}{Self-selected} & \multicolumn{2}{c}{Representative} \\
				& 2 votes & Ranking & 2 votes & Ranking \\\midrule 
			Actual & 7 &  7& 7 & 7\\\midrule 
			DO & 6 & 7$^+$ & 6& 6\\
			STV & 7& 8$^+$ & 7& 7\\
			GP &7& 8$^+$ & 7& 7\\\bottomrule
		\end{tabular}}
	
		\smallskip
		{\footnotesize \emph{Note}: The `+' indicates that a party that did not get \\[-3pt] any seat in the actual election is selected.}
		\caption{Number of selected parties.}
		\label{tab:threshold:results_parties}
	\end{subfigure}
	\caption{Number of unrepresented voters and of selected parties in our datasets when applying different rules.}
	\label{tab:results}
\end{figure}

We first compare the share of voters that are unrepresented (i.e., that did not rank any party that is selected by the rule). \Cref{tab:results_unrepresented} shows these shares obtained from
\begin{enumerate}
    \item[(1)] the actual election (12.1\%),
    \item[(2)] the uninominal rule which deletes everything below voters' first choice, selects all parties meeting the threshold, but leaves voters unrepresented if their first choice is not selected,
    \item[(3)] the party selection rules DO, STV, and GP.
\end{enumerate}

The discrepancy between (1) and (2) is due to voters' strategic behaviour in the actual election. Indeed, as we just saw, many voters, especially in the self-selected sample, put a party that is \emph{out} in first position of their ranking, but voted for a \emph{safe} party at the election. Thus, applying the uninominal rule on the ``sincere'' rankings would lead these voters to be unrepresented in our dataset, while they were represented in the actual election. This partly explains why we observe a much higher percentage of wasted votes for the uninominal rule with our datasets than in the actual election. Another part of the explanation is that, in each dataset, one risky party did not reach the threshold in first-rank votes anymore, because participants voted less strategically, and thus their supporters became unrepresented. 

We are presenting the numbers in (2) as an extreme case of what the effect of the threshold would be if voters were to strategize less. But the numbers are not realistic, since participants decided on their rankings based on an understanding that a party selection rule would be used. Comparing the numbers in (1), (2), and (3), we see that while strategizing is crucial to make the current election work well, it is not necessary when using party selection rules. Indeed, even with the reduced strategic behavior, our rules decrease the share of unrepresented voters even compared to (1).

Let us next compare the performance of the different party selection rules.
Recall that the outcome of DO is always a subset of the outcomes of STV and GP. For each of our datasets, it turns out that STV and GP return the same set of parties, and DO returns one fewer party (see \Cref{tab:threshold:results_parties}). This explains why there are fewer unrepresented voters with STV and GP than with DO. %
The representation gain is higher with unlimited-ranking ballots than with 2-truncated ballots. This is not surprising: when participants rank several parties, they are more likely to rank one that will be selected than when they rank only two. In addition, for the ranking dataset of the self-selected sample, one additional party (namely the \emph{Pirate party}) is selected in the outcome for each rule when evaluated on the rankings, further decreasing the number of unrepresented voters. However, this is almost entirely due to a selection bias: the survey was shared among the supporters and people familiar with this party, leading to more people ranking it first than in a representative sample.\footnote{Note that this bias is only partially corrected by our weighting method, as many participants ranked this party first without indicating voting for it at the election, probably for strategic reasons.}

The representation gain of rankings is higher for the self-selected sample than for the representative sample. This is partly explained by the fact that in the representative sample, a significant number of voters ranked only one party, thus limiting their chances to be represented. This applies to around 40\% of the unrepresented voters in the representative sample, but less than 10\% in the self-selected sample.

\paragraph{Impact of ballot size}
\begin{figure}[t]
    \centering
    \begin{subfigure}[b]{0.4\textwidth}
        \includegraphics[width=\textwidth]{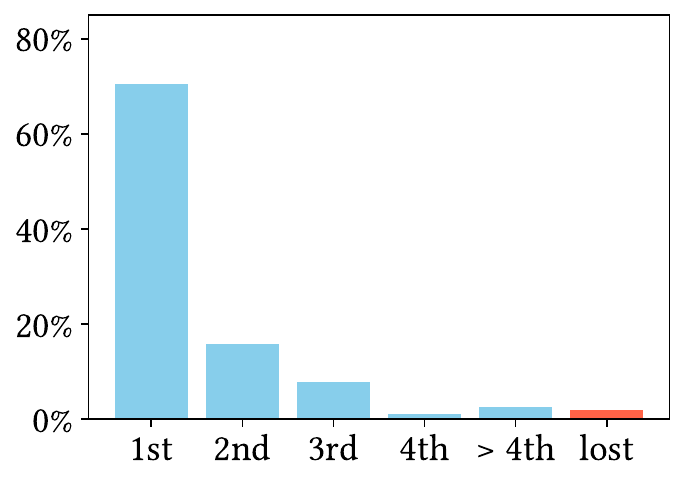}
        \caption{Self-selected sample.}
        \label{fig:rank_representation_autoselected}
    \end{subfigure}
    \qquad
    \begin{subfigure}[b]{0.4\textwidth}
        \includegraphics[width=\textwidth]{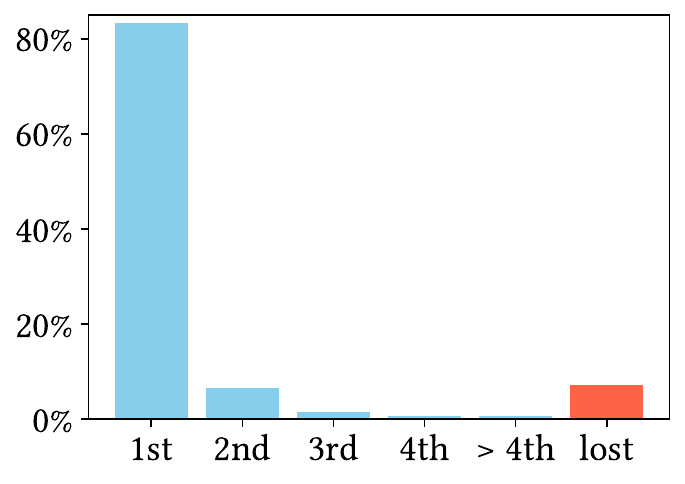}
        \caption{Representative sample.}
        \label{fig:rank_representation_representative}
    \end{subfigure}
    \caption{Distribution of the ranks of the representatives in the rankings of the voters for the STV and GP rules with a 5\% threshold, in the ranking datasets.}
    \label{fig:rank_representation}
\end{figure}

\iflatexml
\begin{figure}
	\includegraphics[width=5.8cm]{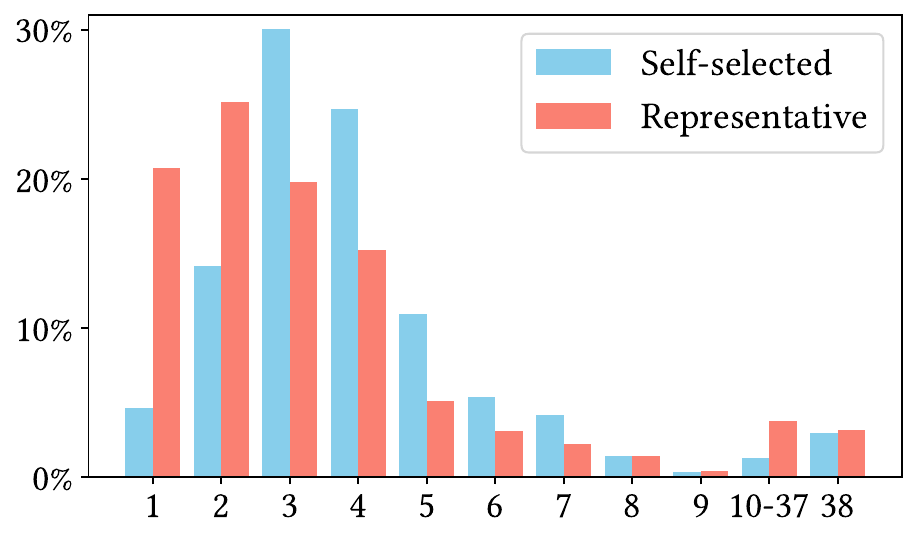}
	\vspace{-8pt}
	\caption{Distribution of how many candidates were ranked by the voters.}
	\label{fig:ranking_lengths}
\end{figure}
\else
\begin{wrapstuff}[r,type=figure,width=6cm,top=4]
	\includegraphics[width=5.8cm]{figs/max_ranking.pdf}
	\vspace{-8pt}
	\caption{Distribution of how many candidates were ranked by the voters.}
	\label{fig:ranking_lengths}
\end{wrapstuff}
\fi
Another interesting observation is that the representativity gain with ranked ballots is already quite high with short ballots. Indeed, a large majority of voters are represented by a party ranked very high in their ranking. For instance, we can see in \Cref{fig:rank_representation_autoselected} that for the ranking dataset of the self-selected sample, if we use STV or GP, around 70\% of voters are represented by their favorite party, and almost all voters have their representative in their top 3 choices. This is even clearer for the representative sample (see \Cref{fig:rank_representation_representative}) as voters in this sample ranked less parties on average (see \Cref{fig:ranking_lengths}).

To complete this analysis, \Cref{fig:unrepresented_ranking_length} shows the fraction of lost votes if all rankings of length greater than $k$ are truncated to rank $k$. For the representative sample, as soon as $k \geq 3$, increasing $k$ has almost no impact on representativity. For the self-selected sample, as voters tend to rank more parties (see \Cref{fig:ranking_lengths}), we need $k \geq 5$ to reach an almost maximal representativity level. This is important for the practical implementation of such rules: limiting voters to rank at most three parties would arguably limit the cognitive load for the voters to an acceptable level (and would thus be more likely to be adopted in real-world political settings) while still ensuring good representativity.

\begin{figure}[!t]
	\begin{subfigure}{.41\textwidth}
		\centering
		\includegraphics[width=\linewidth]{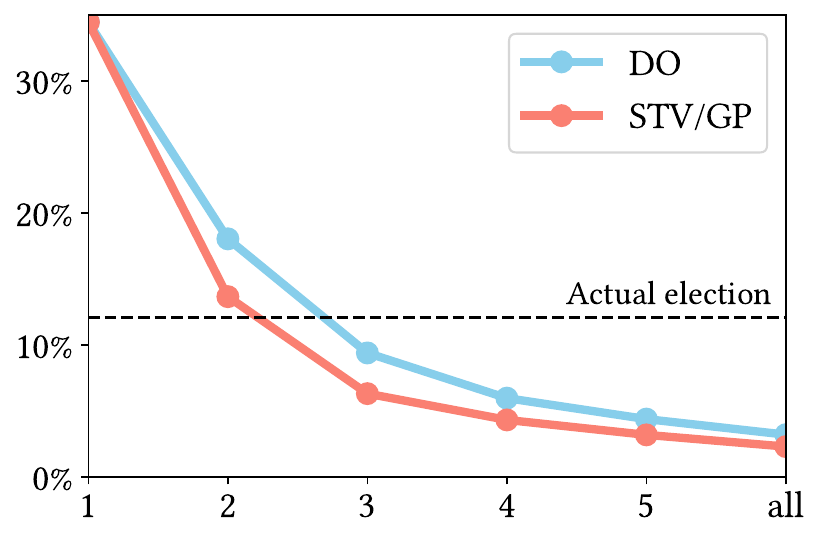}
		\caption{Self-selected sample}
	\end{subfigure}
	\qquad
	\begin{subfigure}{.41\textwidth}
		\centering
		\includegraphics[width=\linewidth]{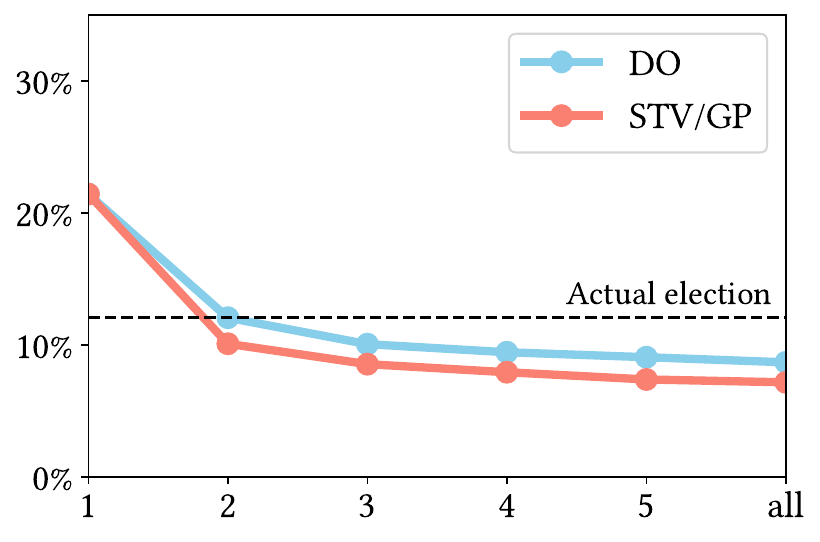}
		\caption{Representative sample}
	\end{subfigure}
	\caption{Percentage of unrepresented voters after truncating their rankings to a particular length, indicated on the horizontal axis.}
	\label{fig:unrepresented_ranking_length}
\end{figure}

\paragraph{Varying the threshold}

\begin{figure}[!p]
    \centering
    \includegraphics[width=\textwidth]{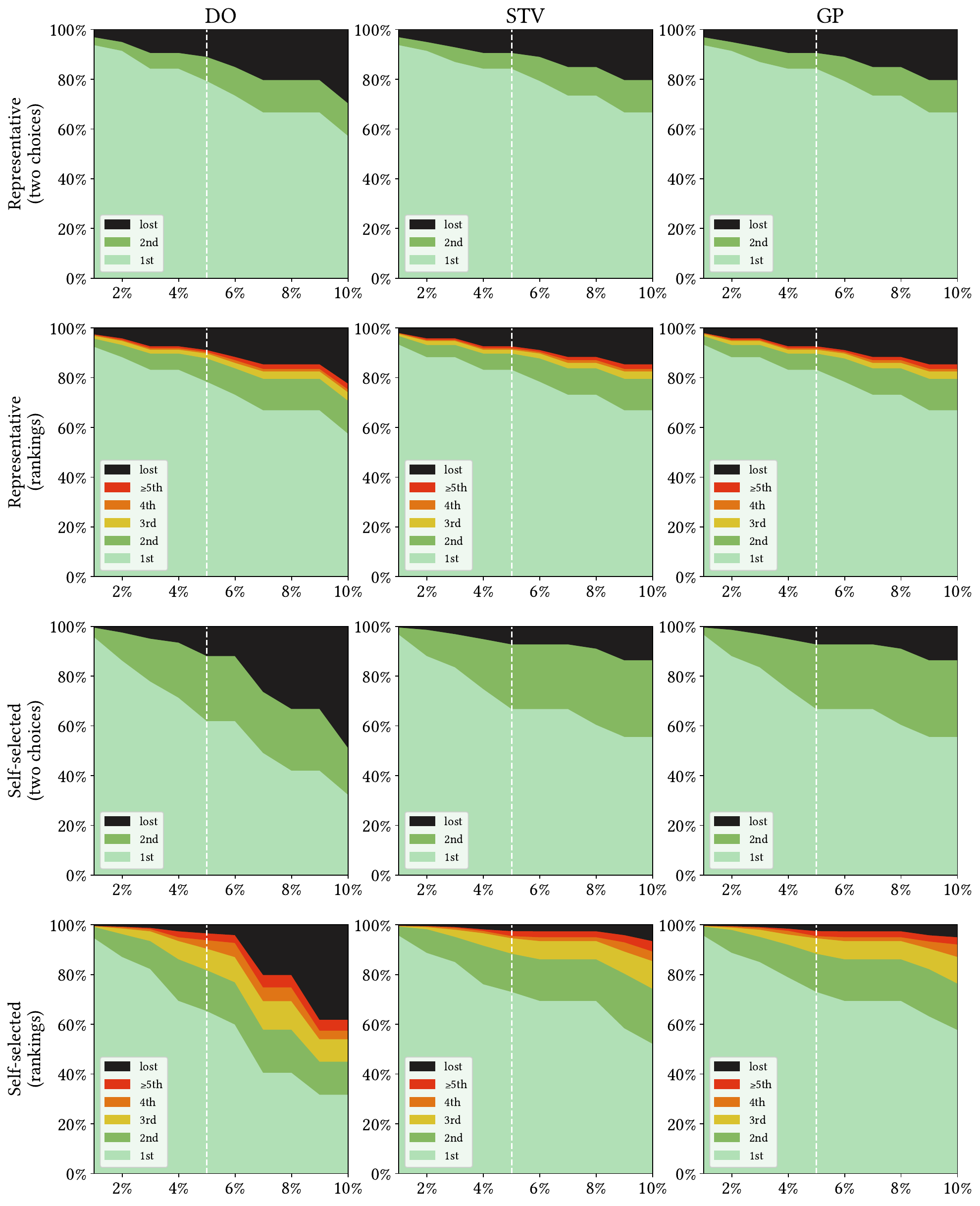}
    \caption{Distribution of the ranks of the representatives in the rankings of the voters with different thresholds. The threshold used is indicated on the horizontal axis. For each threshold, the vertical slice above it shows how the voters are divided into unrepresented voters (black area) and represented voters (colored areas, colored according to the rank of the party assigned to the voter).}
    \label{fig:thresholds_variation}
\end{figure}

The results above were obtained with a threshold of 5\%, the one used in this election. As a robustness check, we ran the analysis with other thresholds. In particular, we applied the rules to our datasets with thresholds varying between 1\% and 10\%, and computed the percentage of unrepresented voters, as well as the distribution of ranks of the representatives. The results are shown in \Cref{fig:thresholds_variation}. They are compatible with our previous observations: STV and GP consistently give similar results and a better representation than DO, especially for high threshold values, for which DO ignores a significant part of the voters. Moreover, we observe that even for high thresholds, most voters are represented by one of their top 3 choices.

\paragraph{Robustness under noise}

Finally, we conducted experiments with random noise added to the data. 
We applied the noise to the weights of the voters, and used the following random model: for every single simulation, we sampled a multiplier for each party $c$ from a Gaussian distribution $\sigma_c \sim \mathcal{N}(1,0.1)$, and one multiplier for each voter $i$ from another Gaussian distribution $\sigma_i \sim \mathcal{N}(1,0.1)$. We then multiplied the original weights $w_c$ of the voters (which depend on their voting intention or actual vote $c \in C$ at the election) by these multipliers. More formally, the new weights are $w^*_i = w_c \times \sigma_c \times \sigma_i$. Thus, the global weights of the different parties can change, and as a result the two risky parties sometimes reach the threshold, and sometimes not. We sampled 100 profiles using this model for each dataset and each rule, and computed the median percentage of unrepresented voters for different values of the threshold, as well as the 20 and 80 percentiles. 

The results are displayed in \Cref{fig:thresholds_variation_noise}. STV and GP continue to give much better results than DO when we add noise to the data. Moreover, we observe that STV and GP are more robust to noise, and that their outcomes are again very similar.

\begin{figure}[!p]
	\centering
	\includegraphics[width=\textwidth]{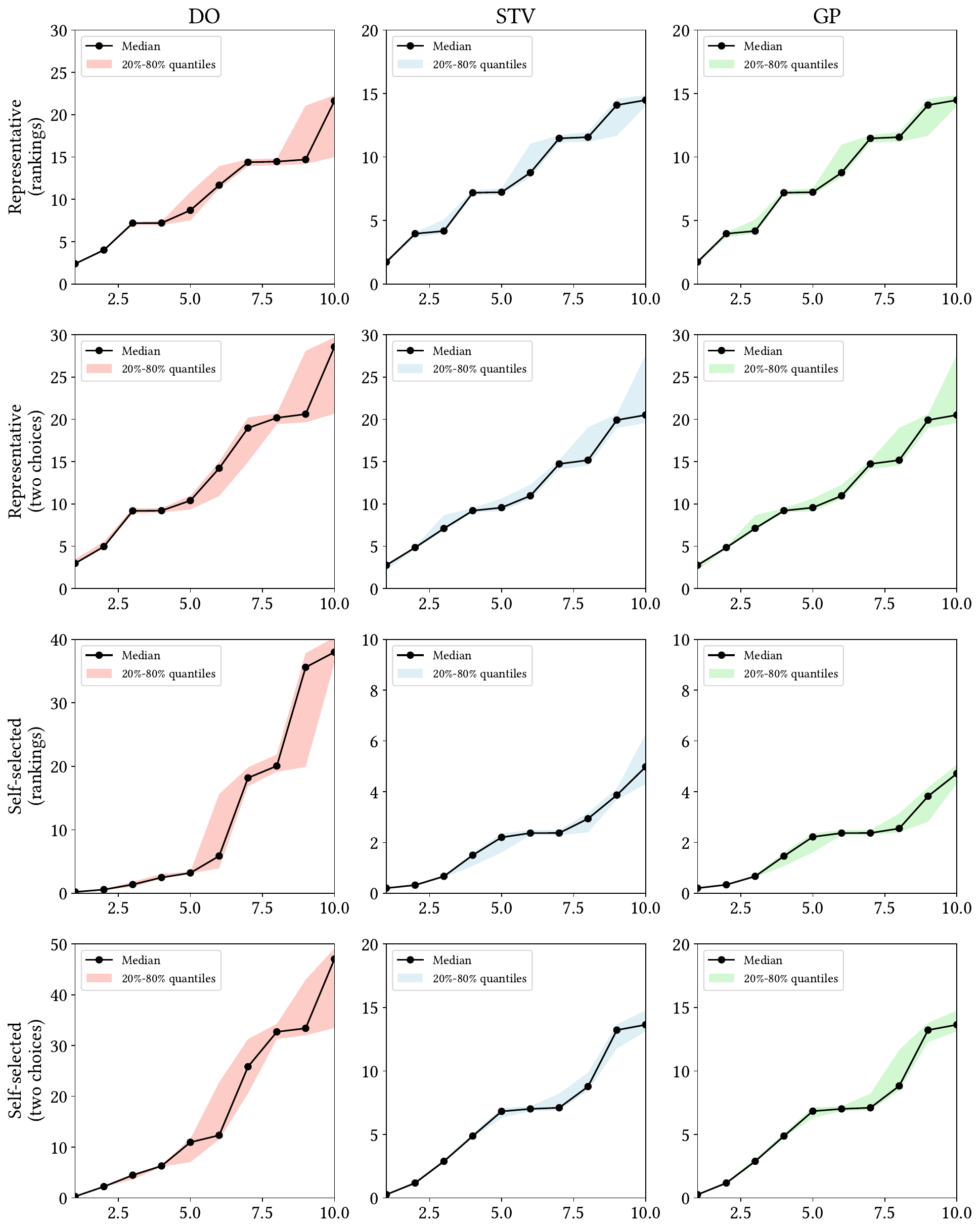}
	\caption{Median percentage of unrepresented voters (over 100 random profiles) for different thresholds (horizontal axis).}
	\label{fig:thresholds_variation_noise}
\end{figure}

\section{Discussion}

We studied whether allowing voters to rank party lists instead of voting for a single list could help obtain more representativeness in parliamentary elections by reducing the amount of unrepresented voters.
Both our theoretical and our empirical results 
suggest that rankings can indeed be helpful, with results varying by rule.
STV and GP allow more parties to be represented, and, relatedly, leave fewer voters unrepresented than DO. Axiomatically, the three rules are incomparable: DO and GP enjoy stronger strategyproofness guarantees than STV, while STV satisfies independence of clones and represents solid coalitions. All three rules are simple and easy to understand, with DO being closest to the current system and STV being closely related to election systems that are widely used, especially in English-speaking countries.

Our theoretical results operate within social choice theory and the axiomatic method. It would be interesting to study this setting from the perspective of strategic candidacy, evaluating the rules' impact on party formation and political innovation. Our experiment could also be productively repeated in other countries, to better understand the robustness of our conclusions, and to understand whether the electorate would accept or welcome the switch to one of these rules.
In addition, studying effects on coalition insurance voting would be interesting.

We hope that our work can support the discussion about proportional representation and thresholds in several countries that are facing issues with their election systems. In our view, the results of our work might even provide reassurance for countries currently using majoritarian first-past-the-post systems to switch to proportional representation. Countries like France that place substantial weight on governability could in principle use methods like the ones we studied to combine proportional representation with very high thresholds (perhaps as high as 10\%) without causing unacceptable amounts of wasted votes.

\bibliographystyle{ACM-Reference-Format}

\clearpage
\appendix

\addtocontents{toc}{\protect\setcounter{tocdepth}{1}}
\setcounter{tocdepth}{1}

\section{Omitted Results of the Axiomatic Analysis} \label{app:axioms}

\subsection{Incentive Issues} \label{app:strategyproofness}

Say that a party selection rule $f$ is a dictatorship for threshold $\tau$ if there exists a voter $i$ such that for every full profile $P$, the output $f(P, \tau) = \{\best_C(i)\}$ (i.e., the outcome is the party that is ranked first by $i$). Say that a party selection rule $f$ is imposing for threshold $\tau$ if there exists some alternative $c_j$ such that $c_j \not \in f(P,\tau)$ for any full profile $P$.

\begin{proposition}
\label{prop:sp-impossibility-large-tau}
For any $m \ge 3$, any representative-strategyproof party selection rule $f$, and any threshold $\tau > \frac{n}{2}$, it must either be the case that $f$ is a dictatorship for threshold $\tau$ or that $f$ is imposing for threshold $\tau$. %
\end{proposition}

\begin{proof}
Note that, by feasibility and since $\tau > \frac{n}{2}$, for every profile $P$ the outcome $f(P, \tau)$ can be only a singleton or the empty set. Suppose that $f( P ,\tau)$ is a singleton for all full profiles $P$; in this case $f( \cdot ,\tau)$ can be seen as a single-winner social choice function that takes full rankings as input. It follows immediately from the Gibbard-Satterthwaite theorem~\citep{gibbard1973manipulation,satterthwaite1975strategy} that if $f( \cdot ,\tau)$ is strategyproof then it is either imposing or a dictatorship for threshold $\tau$. Next, suppose that $f( P , \tau) = \emptyset$ for some full profile $P$. Then there must exist two profiles, $P_1=(\succ_1, \ldots, \succ_i, \ldots, \succ_n)$, and $P_2=(\succ_1, \ldots, \succ_i', \ldots, \succ_n)$, that differ only in the report of a single voter $i$ and for which $f(P_1,\tau) = \emptyset$ and $f(P_2,\tau) \neq \emptyset$.\footnote{If $f(P,\tau)=\emptyset$ for all full profiles $P$ then $f$ is imposing for threshold $\tau$.} Voter $i$ therefore can manipulate by switching their vote from $\succ_i'$ to $\succ_i$.
\end{proof}

We say that a party selection rule is \emph{efficient} if, whenever there are at least $\tau$ non-empty votes and all non-empty votes have the same top choice, then that party should be selected (and, by feasibility, should be the only party selected). We additionally define \emph{anonymity} in the usual way by requiring that a party selection rule treats voters symmetrically. If we strengthen non-imposition to efficiency and non-dictatorship to anonymity, we are able to generalize \Cref{prop:sp-impossibility-large-tau} to all threshold values.

\begin{proposition}
For any $m \ge 3$ and any threshold $\tau>1$, any representative-strategyproof party selection rule $f$ violates either anonymity or efficiency.
\end{proposition}

\begin{proof}
Suppose that there exists a representative-strategyproof, efficient, and anonymous party selection rule $f$ for $m \ge 3$, $n$ voters, and some $1<\tau \le \frac{n}{2}$. Then we can define a representative-strategyproof, non-imposing, and non-dictatorial party selection rule $f'$ for $n'=\tau$ voters. To see this, consider a profile $P'$ consisting of $\tau$ voters. Now, construct a profile $P$ consisting of $n$ voters: the $n'$ voters from profile $P'$ plus $n-\tau$ empty votes. Define $f'(P',\tau)=f(P,\tau)$. Party selection rule $f'$ inherits strategyproofness from $f$ since any successful manipulation under rule $f'$ would also be a successful manipulation under $f$. Non-imposition of $f'$ follows from efficiency of $f$, since efficiency dictates that whenever all non-empty voters have the same first choice, then that party is selected. Finally, non-dictatorship of $f'$ follows from anonymity of $f$: if $f'$ was dictatorial then swapping the roles of a pair of voters would sometimes change the outcome, which would violate anonymity of $f$.
\end{proof}

\section{Optimization Rules}
\label{app:max-rules}

In this appendix section, we define two optimization-based rules and analyse their computational complexity and some of their axiomatic properties.

\subsection{Definition}

\begin{itemize}
	\item \emph{Maximal Plurality (MaxP):} This rule return the set $S$ of parties that is feasible and that maximizes the number of voters that rank at least one party of the set first. In case of ties, the rule maximizes the number of voters that rank at least one party of the set in the first two positions, and so on. %
	\item \emph{Maximal Representation (MaxR):} This rule return the set $S$ of parties that is feasible and that maximizes the number of voters that include at least one party from the set in their ranking (at any position). %
\end{itemize}

If we have to break ties between sets of parties, we do it by selecting the set that is lexicographically maximal according to the fixed order on the parties. 

Using the profile from \Cref{ex:example_diff}, the MaxP rule will return $\{d,a\}$, with 10 voters having their top choice selected. The MaxR rule will select $\{d,b\}$, with all voters being represented by one of their top two choices.

\subsection{Computational Complexity}
 We can show that the problem of computing the outcome of MaxP and MaxR is NP-hard, as we need to optimize over all possible sets of parties.

\begin{theorem}
	The problem of computing the outcome of MaxP and MaxR is NP-hard.
\end{theorem}
\begin{proof}
	We prove it by reduction of the independent set problem. Given a 3-regular graph $G$ (in which every vertex has degree 3), the problem is to find a set of vertices of maximal size such that no two vertices are connected.

	Let $G = (U, E)$ be 3-regular graph. Construct a profile $P$ with party set $C = U$ and such that for each $(u,v) \in E$, we add one voter with preference $u \succ v$ and one voter with preference $v \succ u$. We set the threshold $\tau = 6$. We claim that there is an independent set of size $k$ in $G$ if and only if there is a feasible outcome of size $k$ in $P$. Moreover, an outcome is feasible if and only if no two vertices of it are connected by an edge (as there are exactly six voters ranking each vertex). Thus, since every vertex appears in 6 rankings, and 3 times in first position, then the outcome of maximal size is necessarily the outcome of MaxP and MaxR.

	Now, assume that there exists an independent set of size $k$ in $G$, and note it $S \subseteq U$. Then, since vertices in $S$ are independent, no voter in $P$ rank two vertices of $S$. Thus, every vertex represents six voters, and the outcome $S$ is feasible. Conversely, if the outcome if feasible, this means that no voter rank two vertices of the outcome, thus there are no edges between the vertices of the outcome, and the outcome is an independent set of $G$. 

	Therefore, the problem of computing the outcome of MaxP and MaxR is NP-hard.
\end{proof}

In addition to being hard to compute, the MaxP and MaxR rules are probably not appropriate to be used in practice for parliamentary election, as the outcome, and how it was selected, can hardly be explained. However, they can be useful in other contexts than parliamentary elections, such as group activity selection or food item selection, where it can be more important for voters to be satisfied with the outcome than to understand how it was selected.

\subsection{Axiomatic Properties}

\begin{proposition}
	MaxP and MaxR satisfy set-maximality.
\end{proposition}
\begin{proof}

	The result is clear since we are maximizing some form of representation among voters. So if there exists a superset of the outcome that is feasible, that means that we can increase the representation of some voters, which is a contradiction.
\end{proof}

Perhaps surprisingly, MaxP and MaxR fail inclusion of direct winners. Indeed, consider for instance the profile $ P = \{2: a\succ b, 2:a\succ c, 3: c, 3: b \}$ with $\tau = 4$. The only direct winner is $a$, however both MaxP and MaxR will return $\{c,b\}$. A direct consequence is that these rules fail representation of solid coalitions.

\begin{proposition}
	MaxP and MaxR fail threshold monotonicity.
\end{proposition}
\begin{proof}
	For MaxP and MaxR, consider the profile $P = \{5:a, 1:a \succ b, 1:a \succ c, 4: b , 4: c \}$. With $\tau = 5$, the outcome is $\{b,c\}$ and with $\tau' = 7$, the outcome is $\{a\}$.
\end{proof}

\begin{proposition}
	MaxP and MaxR fail monotonicity.
\end{proposition}
\begin{proof}
	Consider the profile $P = \{5: b, 1:b \succ c, 5:c, 1:c\succ b, 3: a \succ b, 3: a \succ c, 4:a\}$ with $\tau = 9$. If $a$ is in the outcome, then neither $b$ nor $c$ can be because they will represent at most 8 voters. Thus, the two possible outcomes of maximal size are $\{a\}$ and $\{b,c\}$, and $\{b,c\}$ is the one maximizing both the number of voter represented and the number of voters with a first choice. Now, consider the profile $P' = \{5: b, 1:b \succ c, 5:c, 1:\underline{b}\succ \underline{c}, 3: a \succ b, 3: a \succ c, 4:a\}$ in which we increased the ranking of $b$ in a ranking. In this profile, $a$ is still incompatible with $b$ and $c$, but now $b$ and $c$ are also incompatible together, as $c$ will only represent 8 voters. Thus, the only possible outcomes of maximal size are singletons, and $\{a\}$ is the one that maximizes the objectives of MaxP and MaxR. This contradicts monotonicity.
\end{proof}

\begin{proposition}
	MaxR satisfies independence of clones, MaxP does not.
\end{proposition}
\begin{proof}

	For MaxR, simply observe that the number of voters represented by a set $S \subseteq C$ in $P'$ is the same than in $P$, and the number of voters represented by $S \cup \{c'\}$ in $P'$ is the same than in $P$ if $c \in S$, and is equal to the one of $S \cup \{c\}$ otherwise. Thus, MaxR will satisfy the conditions of the axiom.
	
	For MaxP, we can use the same profile than for GP and DO.
\end{proof}

\section{Representation of Unrepresented Voters}\label{app:unrepresented}

In this Appendix section, we discuss the representation of unrepresented voters axiom, show that it is not satisfied by DO, GP, and STV, and propose variants of the rules that satisfy it. Then, we check which axioms are satisfied by these variants.

\subsection{Failure of the Axiom}

\begin{proposition}
	\label{prop:representation_unrepresented}
	DO, GP, STV, MaxP and MaxR do not satisfy representation of unrepresented voters.
\end{proposition}
\begin{proof}
	To show that DO and STV fail this axiom, consider the following profile $P= \{3:a, 2:b\succ c, 1: c\}$ with $\tau = 3$. In this profile, $\DO(P,\tau) = \STV(P,\tau) = \{a\}$. But the last three voters all ranked $c$ and are unrepresented, thus breaking the axiom.

	For GP, consider the profile $P = \{3:a, 2:b \succ a, 2: c \succ b, 2: d \succ b\}$ with $\tau = 4$. The first party added to the outcome is $a$, since it is ranked first by the most voters. Then, $c$ and $d$ cannot be added because they are ranked by only two voters each, and $b$ cannot be added because it will reduce the number of supporters of $a$ to only three, which is below the threshold, making the set $\{a,b\}$ unfeasible. Thus, $GP(P,\tau) = \{a\}$, but the last $\tau = 4$ voters all ranked $b$ and are unrepresented, breaking the axiom.

\end{proof}

\subsection{Variants of the Rules}

In this section, we describe the algorithm to generate outcomes satisfying the representation of unrepresented voters axiom. In the following, we write $c \succ \emptyset$ to mean that $c$ is ranked in the truncated ranking $\succ$. The algorithm is as follows:

\begin{enumerate}
	\item Start with any feasible set of parties $S_0$ (e.g., the outcome of DO, STV or GP).
	\item Now, repeat the following until there is no party $c$ such that $|\{ i \in N : \best_{S_k}(i) = \emptyset \text{ and } c\succ_i \emptyset\}| \ge \tau$:
	\begin{enumerate}
		\item Let $c$ be a party such that $|\{ i \in N : \best_{S_k}(i) = \emptyset \text{ and } c\succ_i \emptyset\}| \ge \tau$.
		\item Identify all parties $c'$ which do not get enough support anymore if we add $c$. More formally, this corresponds to the set $S_{-}$ such that $c' \in S_{-}$ if  $|\{ i \in N : \best_{S_k}(i) = c' \text{ and } c' \succ c\}| < \tau$. %
		\item Add $c$ to the outcome, and remove all parties from $S_{-}$: $S_{k+1} = S_k \cup \{c\} \setminus S_{-}$.
	\end{enumerate}
\end{enumerate}

We can show that this rule terminates. If it terminates, this means we cannot find any party satisfying the condition of the axiom in step (2), and thus the representation of unrepresented voters axiom is satisfied. Note that we left some freedom in the rule, as we can choose the initial set of parties $S_0$ and the way we break ties in the selection of the party to add. We denote DO$^+$, STV$^+$ and GP$^+$ the rules that are obtained by starting with the outcome of DO, STV, and GP respectively, and breaking ties by selecting the party that is supported by the most unrepresented voters $c = \arg\max_{c' \notin S_k} |\{ i \in N : \best_{S_k}(i) = \emptyset \text{ and } c'\succ_i \emptyset\}|$.

\begin{proposition}
	The algorithm described above terminates.
\end{proposition}
\begin{proof}
To see why the algorithm that successively adds parties to the outcome terminates, observe that if it does not, this means that it cycles. Thus, if we denote $S_k$ the outcome obtained at the step $k$, then there exists $k' > k$ such that $S_{k'} = S_k$. For a subset $S \subseteq C$, denote $r(S)$ the number of voters who ranked at least one alternative from $S$:
\[
r(S) = |\{ i \in N \mid best_S(i) \ne \emptyset \}
\]

We have $r(S_{k'}) = r(S_k)$. Note that at each step we add a party to the outcome such that there exists at least $\tau$ unrepresented voters who have ranked this party. If this cause some other parties to be removed from the outcome, this means that these parties only had strictly less than $\tau$ supporters putting them above all other parties from the outcome. Thus, by removing a party, we add at most $\tau-1$ unrepresented voters. Therefore, we have the following for all $j$:
\begin{align*}
	r(S_{j+1}) \ge r(S_j) + \tau - (|S_j|+1 - |S_{j+1}|)\cdot(\tau-1) \\
	r(S_{j+1}) -  r(S_j) \ge 1 - (|S_j| - |S_{j+1}|)\cdot(\tau-1)
\end{align*}
Thus, we have 
\begin{align*}
	r(S_{k'}) -  r(S_k) 
=&\sum_{j=k}^{k'-1} r(S_{j+1}) -  r(S_j) 
\ge \sum_{j=k}^{k'-1} 1 - (|S_j| - |S_{j+1}|)\cdot(\tau-1) \\
=&(k'-1-k) - |S_k| + |S_{k'}| 
=(k'-1-k)
\end{align*}
since $S_k = S_{k'}$. Moreover, we clearly have $k' > k+1$, otherwise this means that $S_{k+1} = S_k$, which is impossible since we add a new party at each step. Thus, we have $r(S_{k'}) -  r(S_k) \ge 1$, which is a contradiction with the fact that $r(S_{k'}) = r(S_k)$. Therefore, there exists no cycle and the algorithm terminates.
\end{proof}

\begin{corollary}
	DO$^+$, STV$^+$ and GP$^+$ satisfy representation of unrepresented voters.
\end{corollary}

Note moreover that these rules also satisfy inclusion of direct winners, and that STV$^+$ satisfies representation of solid coalitions, as it is not possible to eliminate all parties of a solid coalition (since they have at least $\tau$ supporters, there will always be at least one party of the coalition in the outcome). However, DO$^+$ and GP$^+$ do not satisfy representation of solid coalitions, since the counter example provided in the proof is a full profile. However, note that these variants will fail some axioms satisfied by the original rules. In particular, they fail all the following axioms, except independence of clones which is satisfied by STV$^+$.

\section{Screenshots} \label{app:screenshots}
In this appendix, we display screenshots of our survey, together with a version translated into English by Google Translate. The screenshots are shown in Figures \ref{fig:page1} to \ref{fig:page8}. Page 3 is omitted to save space; it shows the posters of all 38 lists.

\begin{figure}[h]
	\includegraphics[width=0.48\linewidth,frame=0.2pt,trim=1cm 12.5cm 1cm 1cm,clip]{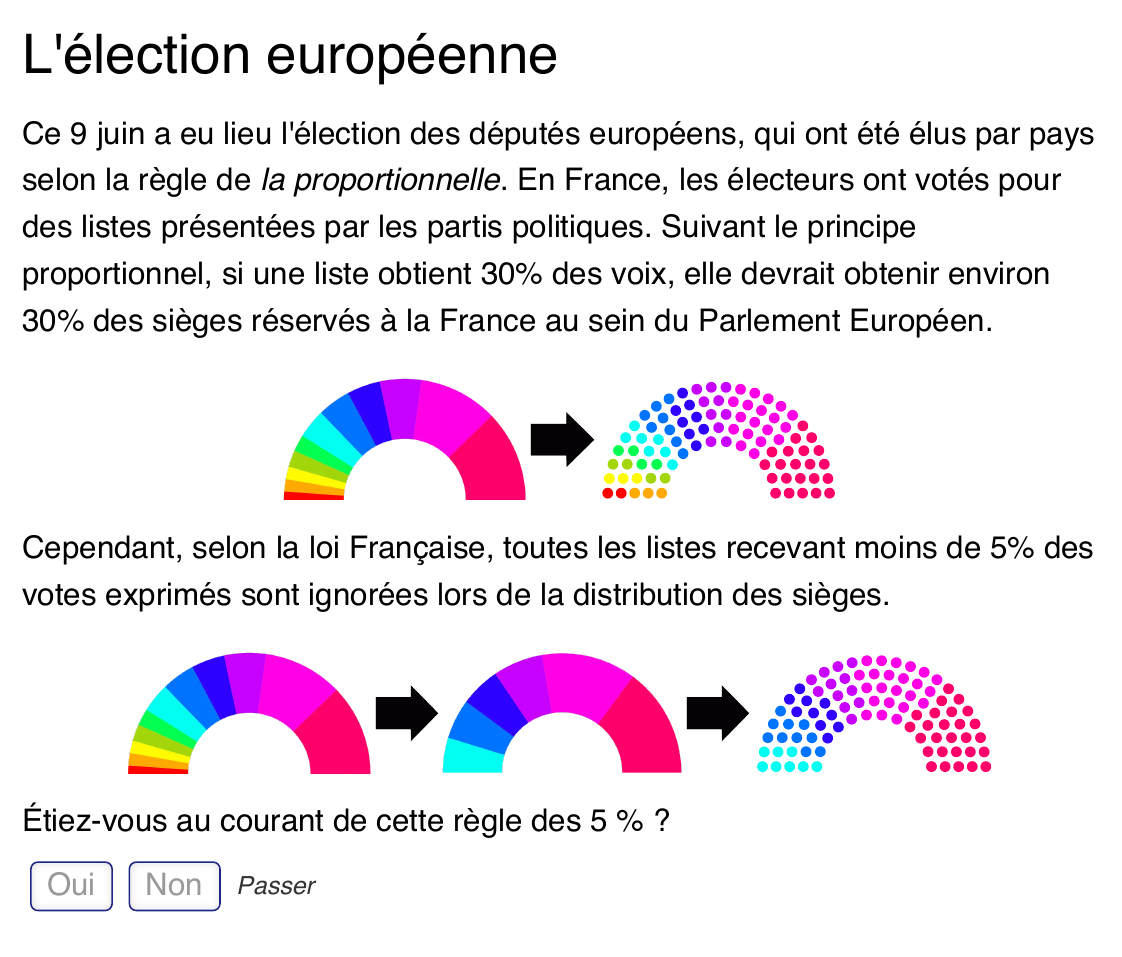}
	\includegraphics[width=0.48\linewidth,frame=0.2pt,trim=1cm 12.5cm 1cm 1cm,clip]{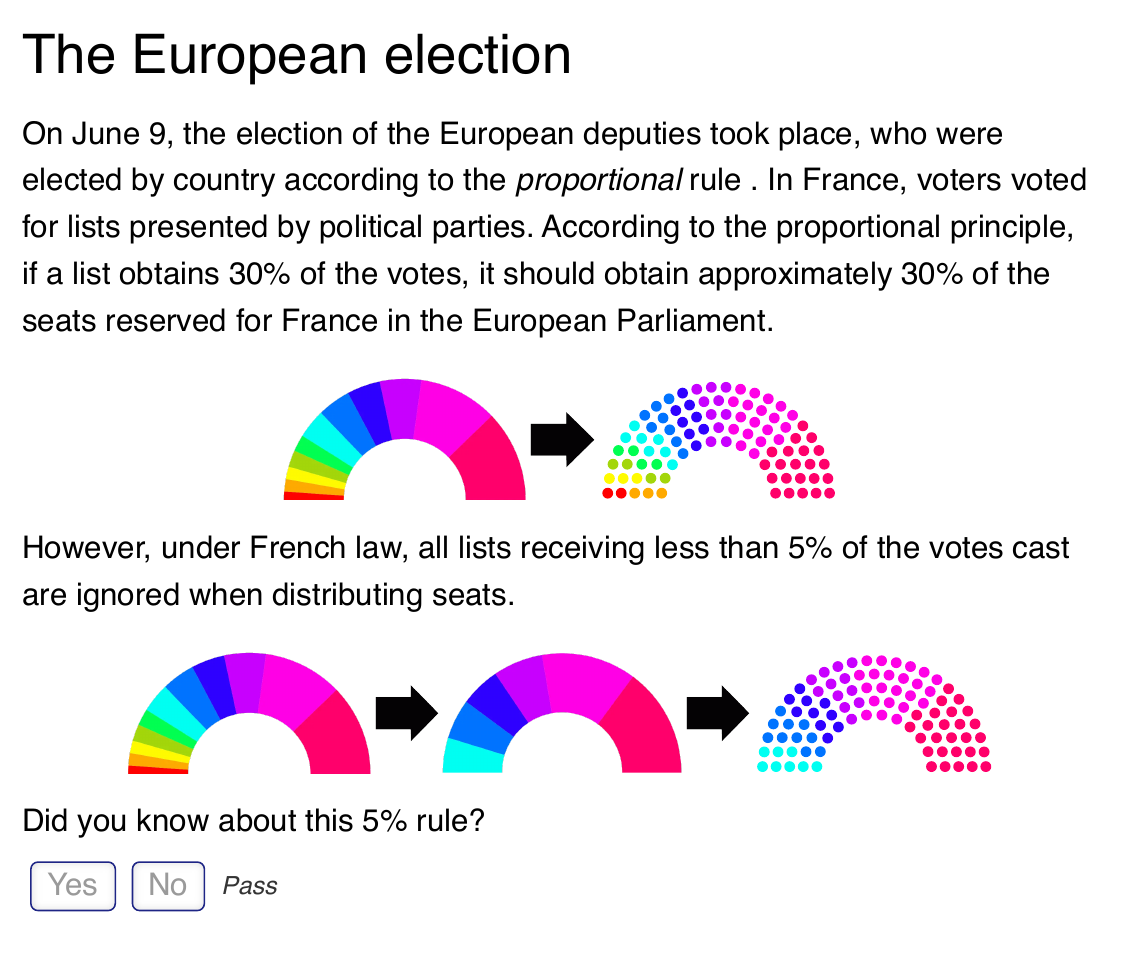}
	\caption{Page 1}
	\label{fig:page1}
\end{figure}

\begin{figure}[h]
	\includegraphics[width=0.48\linewidth,frame=0.2pt,trim=1cm 1cm 1cm 1cm,clip]{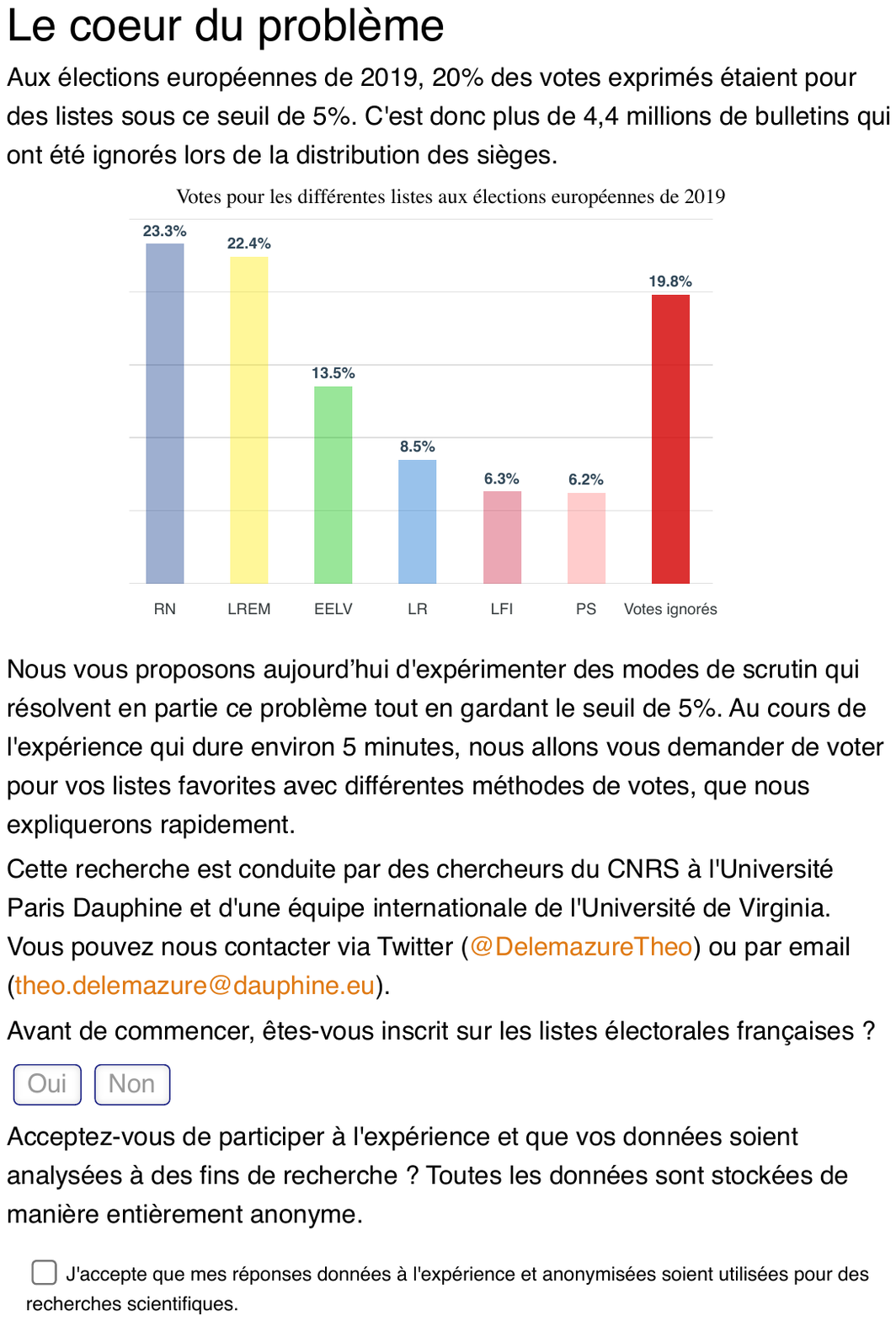}
	\includegraphics[width=0.48\linewidth,frame=0.2pt,trim=1cm 1cm 1cm 1cm,clip]{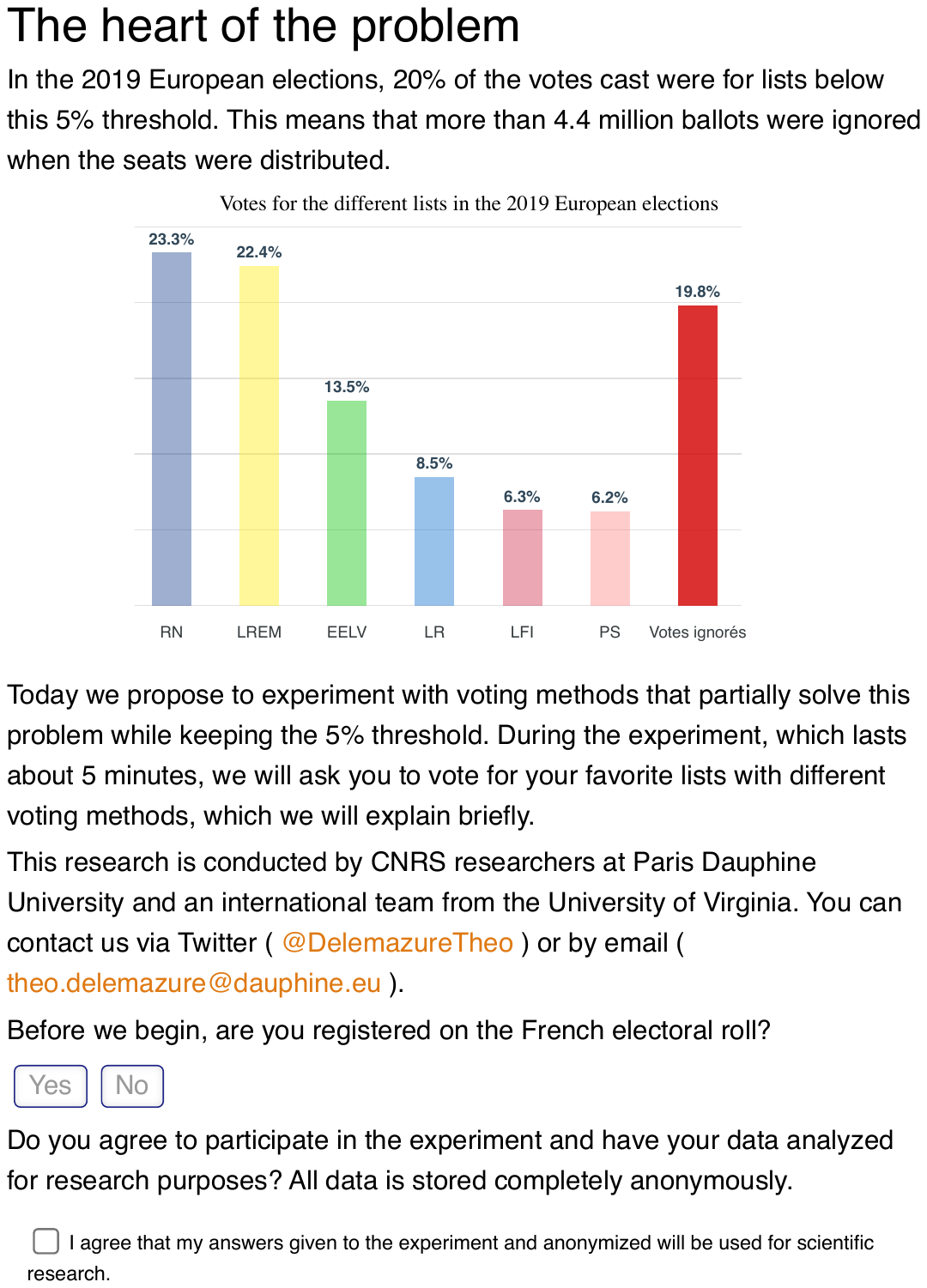}
	\caption{Page 2}
	\label{fig:page2}
\end{figure}

\begin{figure}[h]
	\includegraphics[width=0.48\linewidth,frame=0.2pt,trim=1cm 5cm 1cm 1cm,clip]{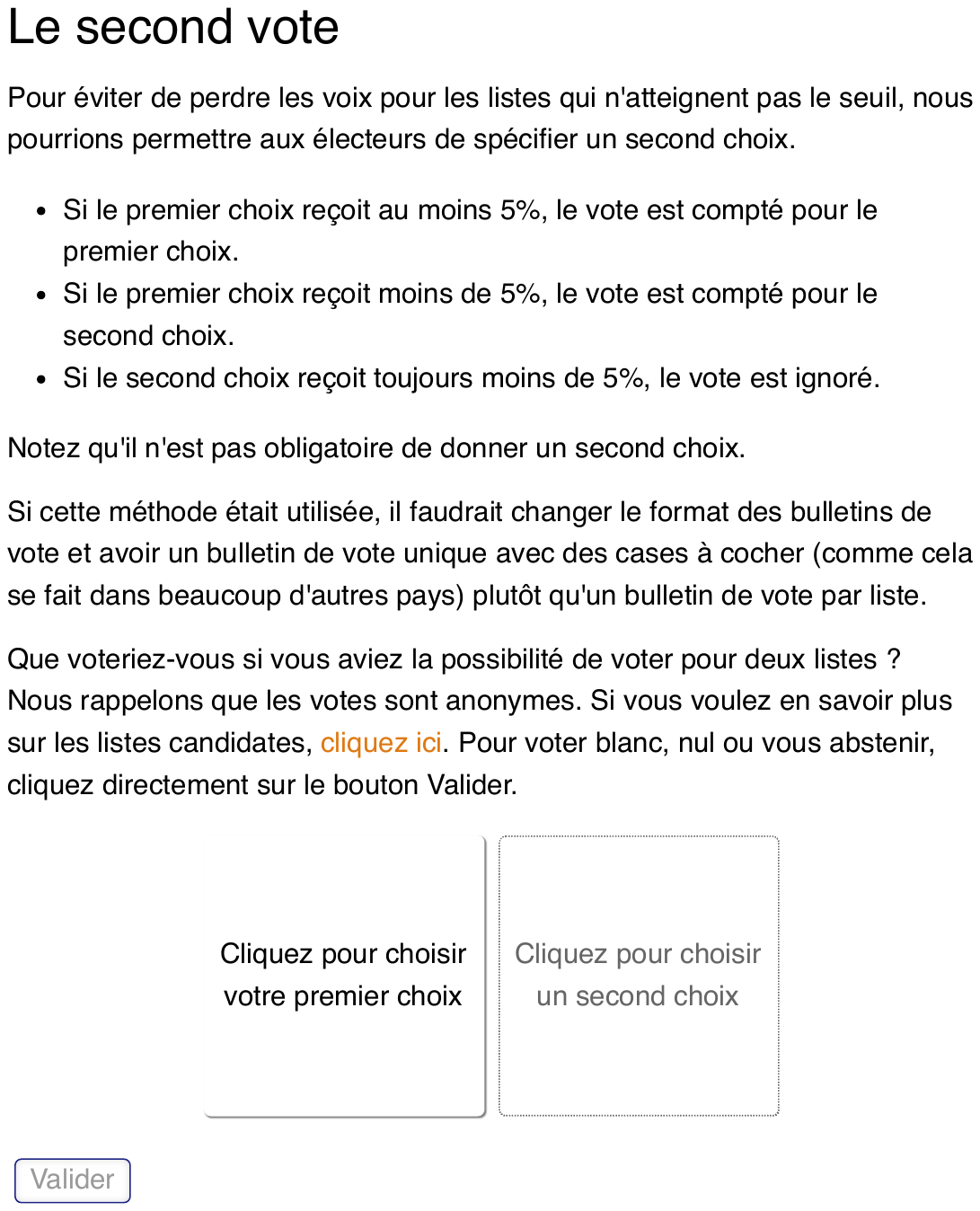}
	\includegraphics[width=0.48\linewidth,frame=0.2pt,trim=1cm 5cm 1cm 1cm,clip]{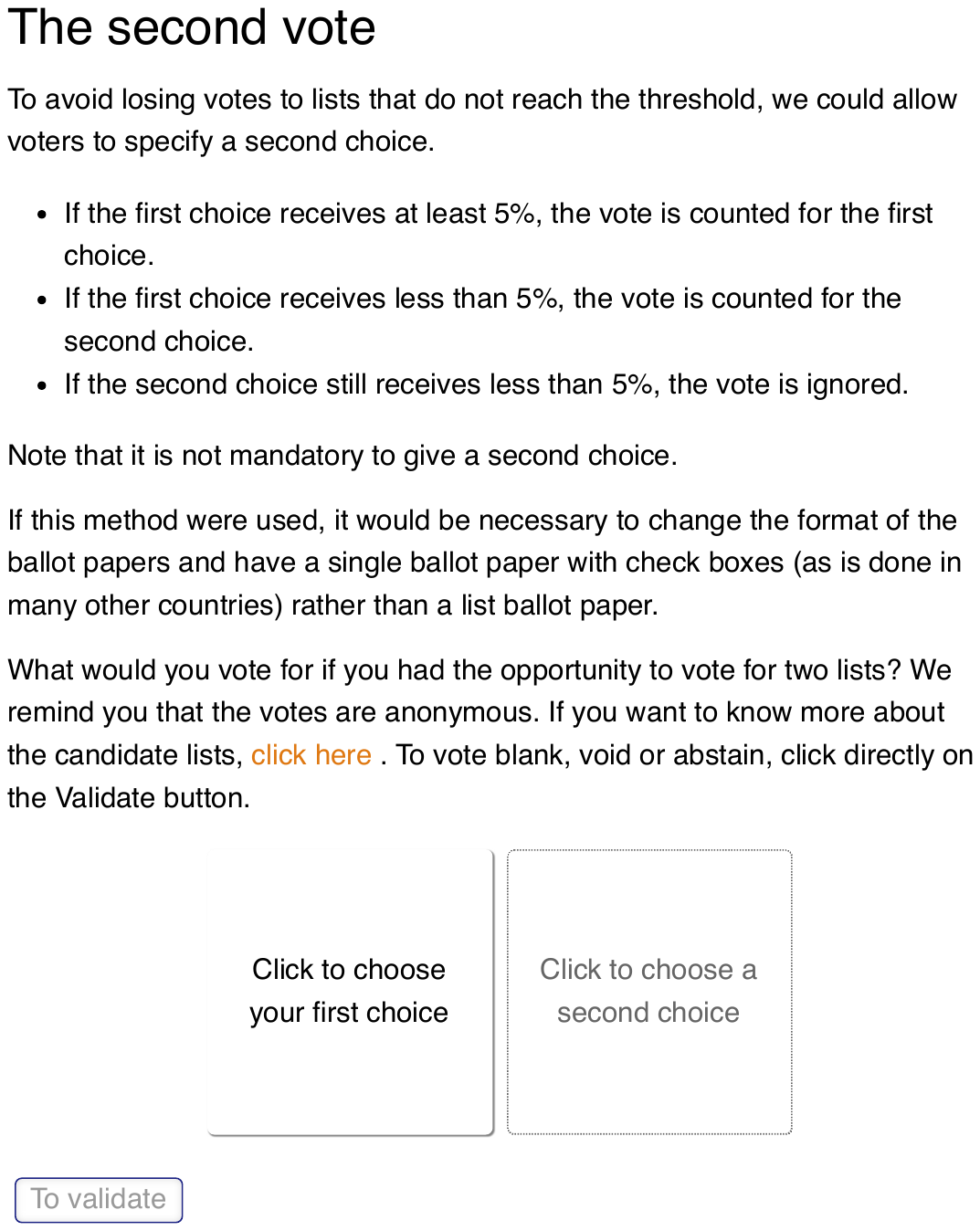}
	\caption{Page 4}
	\label{fig:page4}
\end{figure}

\begin{figure}[h]
	\includegraphics[width=0.48\linewidth,frame=0.2pt,trim=1cm 5.2cm 1cm 1cm,clip]{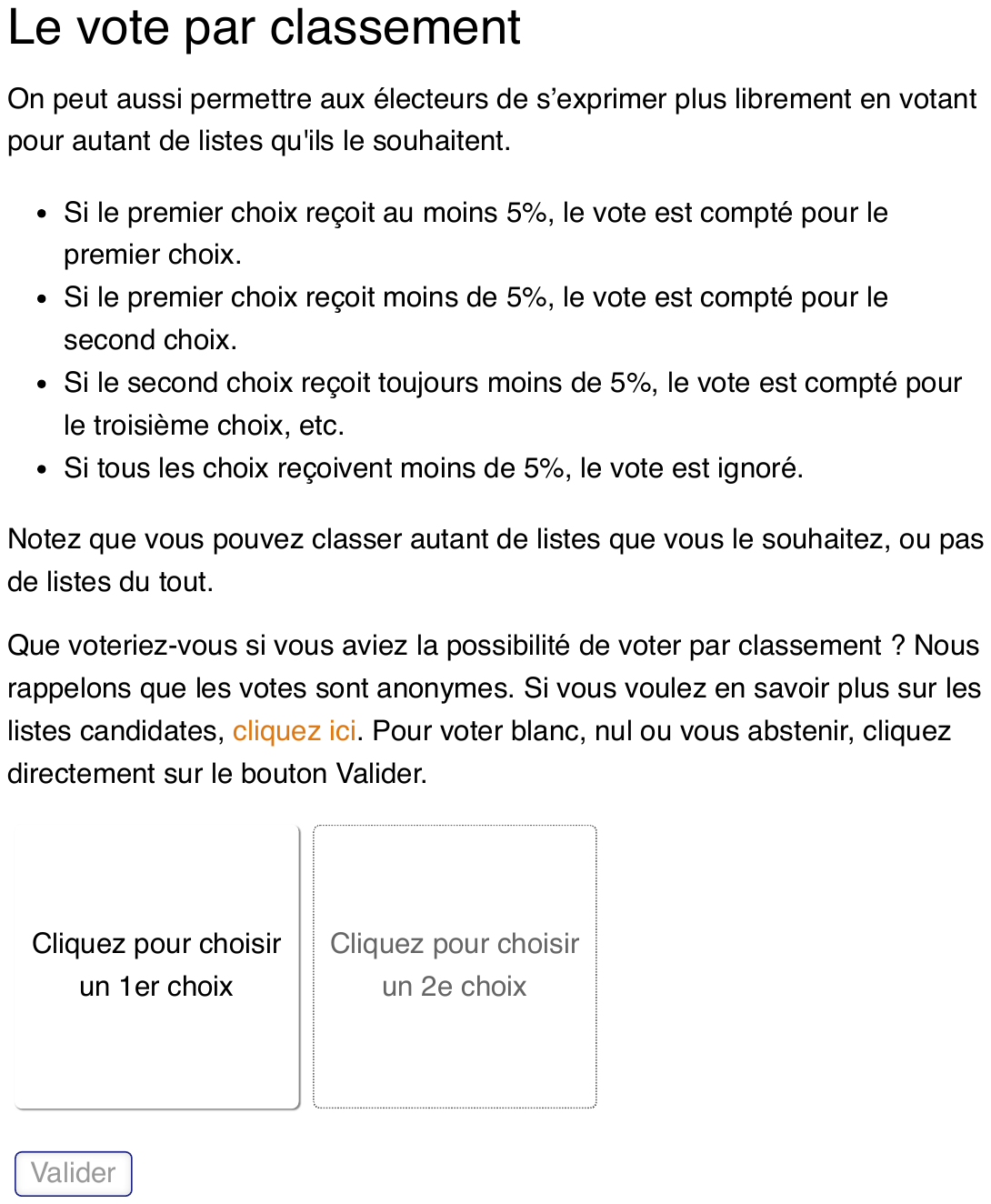}
	\includegraphics[width=0.48\linewidth,frame=0.2pt,trim=1cm 5.2cm 1cm 1cm,clip]{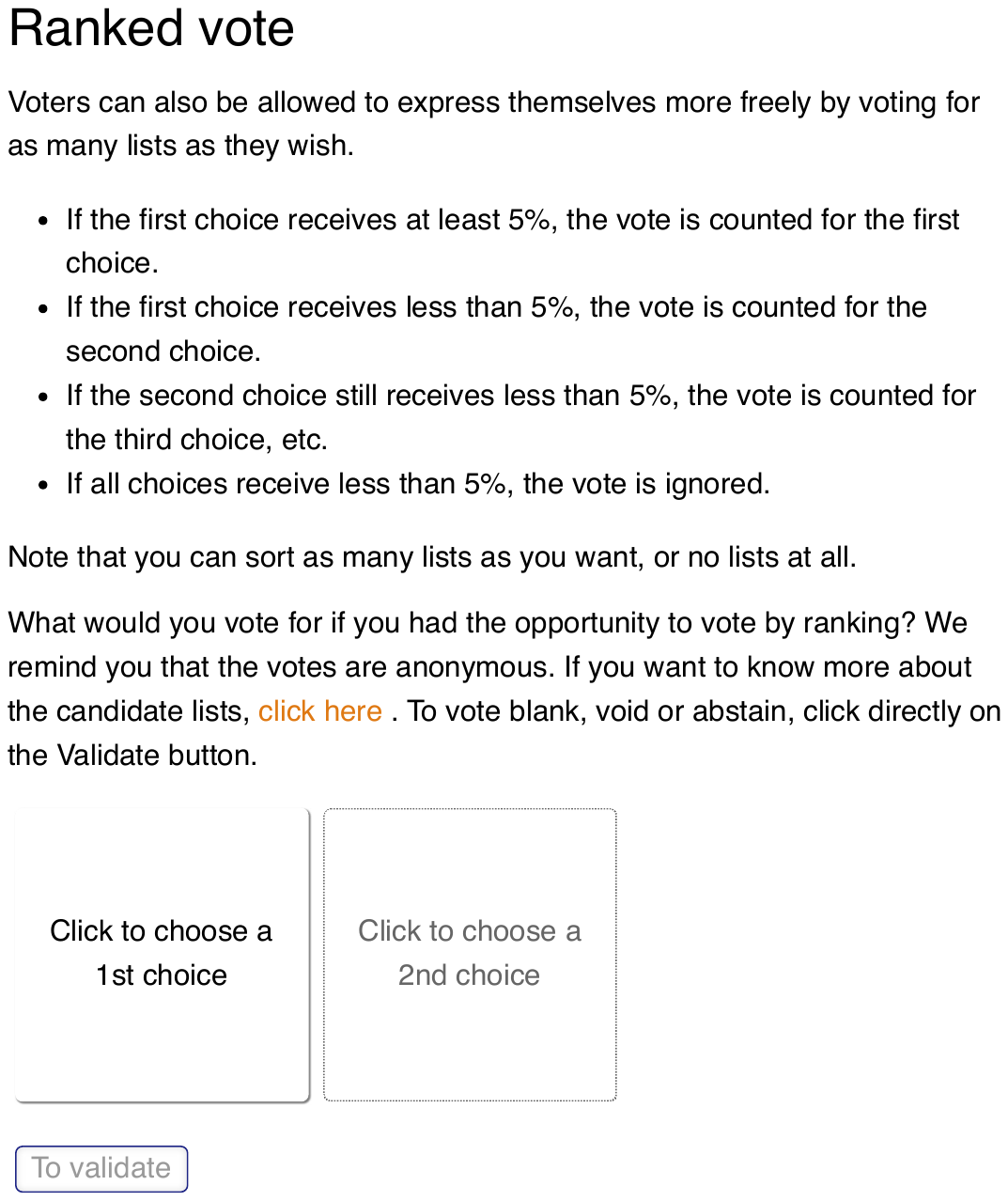}
	\caption{Page 5}
	\label{fig:page5}
\end{figure}

\begin{figure}[h]
	\includegraphics[width=0.48\linewidth,frame=0.2pt,trim=1cm 14.5cm 1cm 1cm,clip]{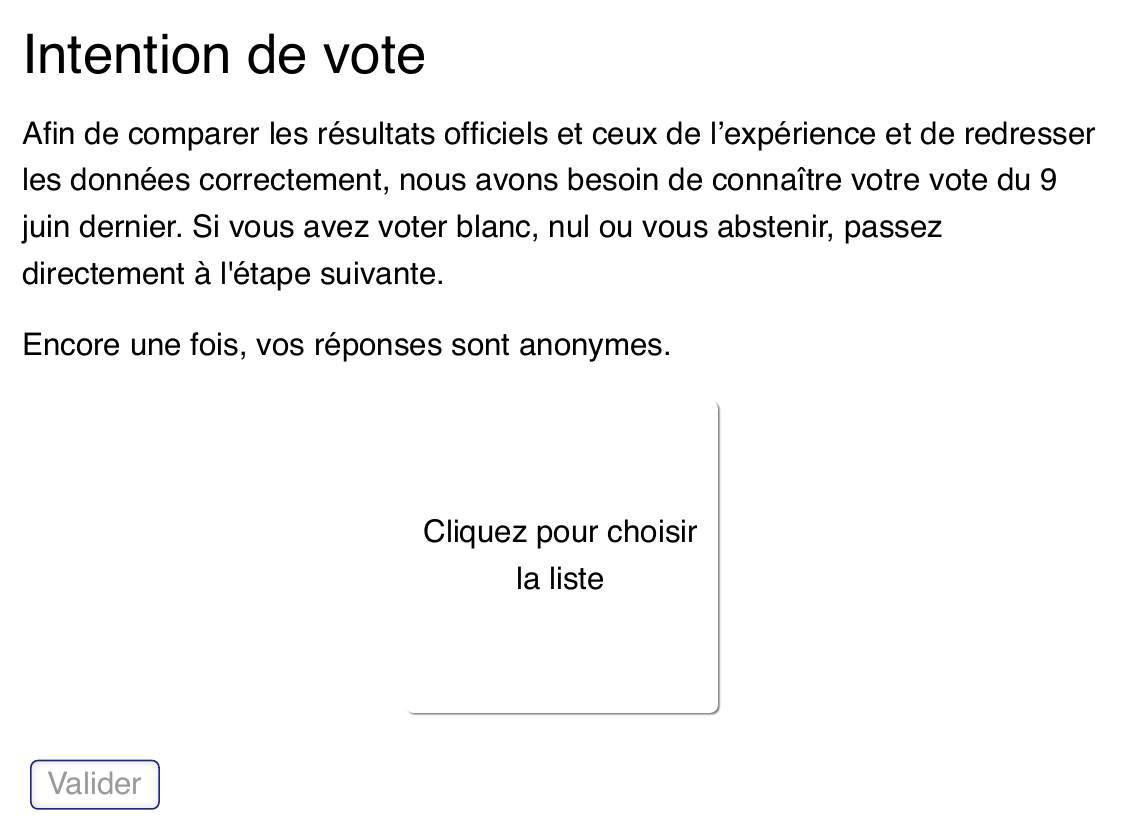}
	\includegraphics[width=0.48\linewidth,frame=0.2pt,trim=1cm 14.5cm 1cm 1cm,clip]{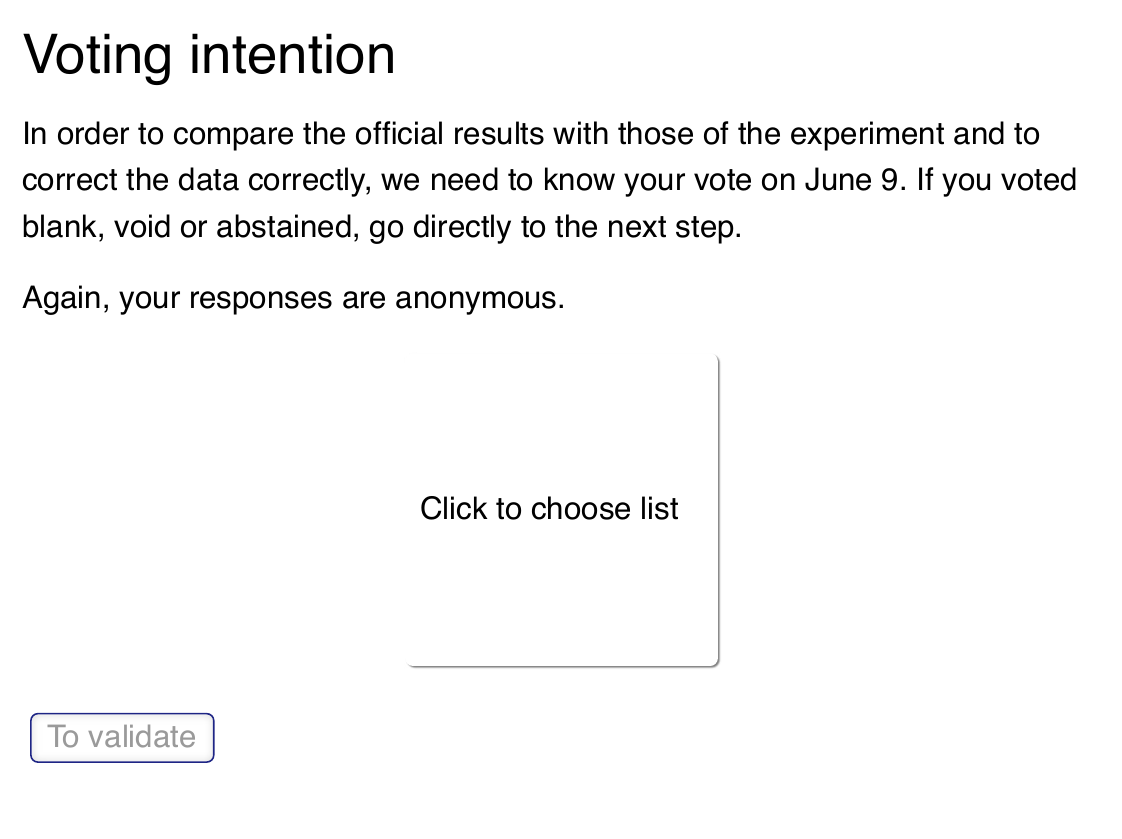}
	\caption{Page 6}
	\label{fig:page6}
\end{figure}

\begin{figure}[h]
	\includegraphics[width=0.48\linewidth,frame=0.2pt,trim=1cm 15cm 1cm 1cm,clip]{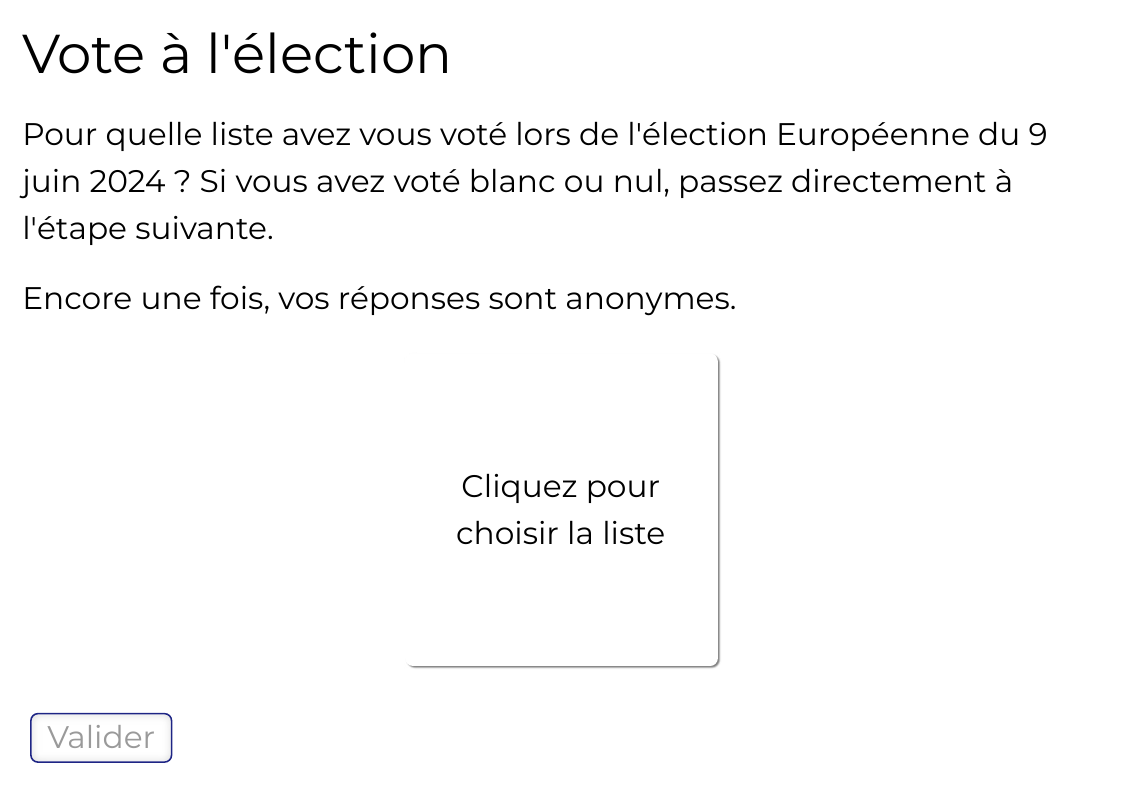}
	\includegraphics[width=0.48\linewidth,frame=0.2pt,trim=1cm 15cm 1cm 1cm,clip]{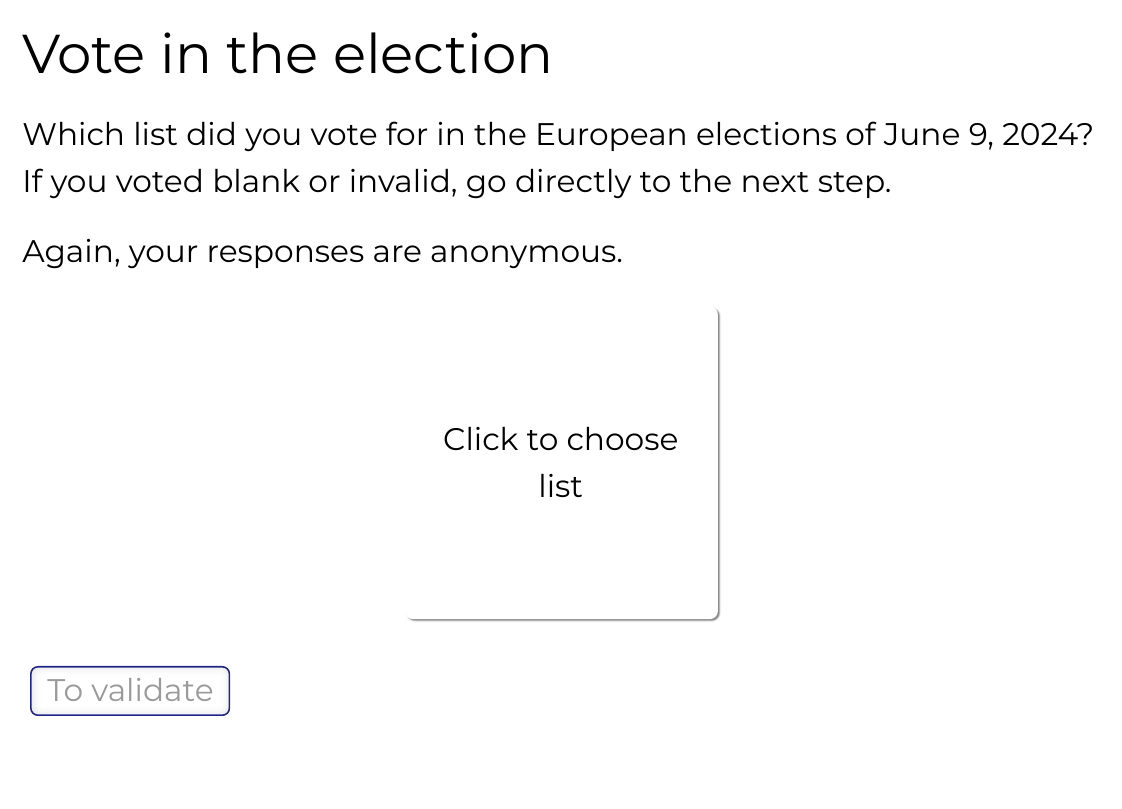}
	\caption{Page 6 (after election day)}
	\label{fig:page6dynata}
\end{figure}

\begin{figure}[h]
	\includegraphics[width=0.48\linewidth,frame=0.2pt,trim=1cm 5cm 1cm 1cm,clip]{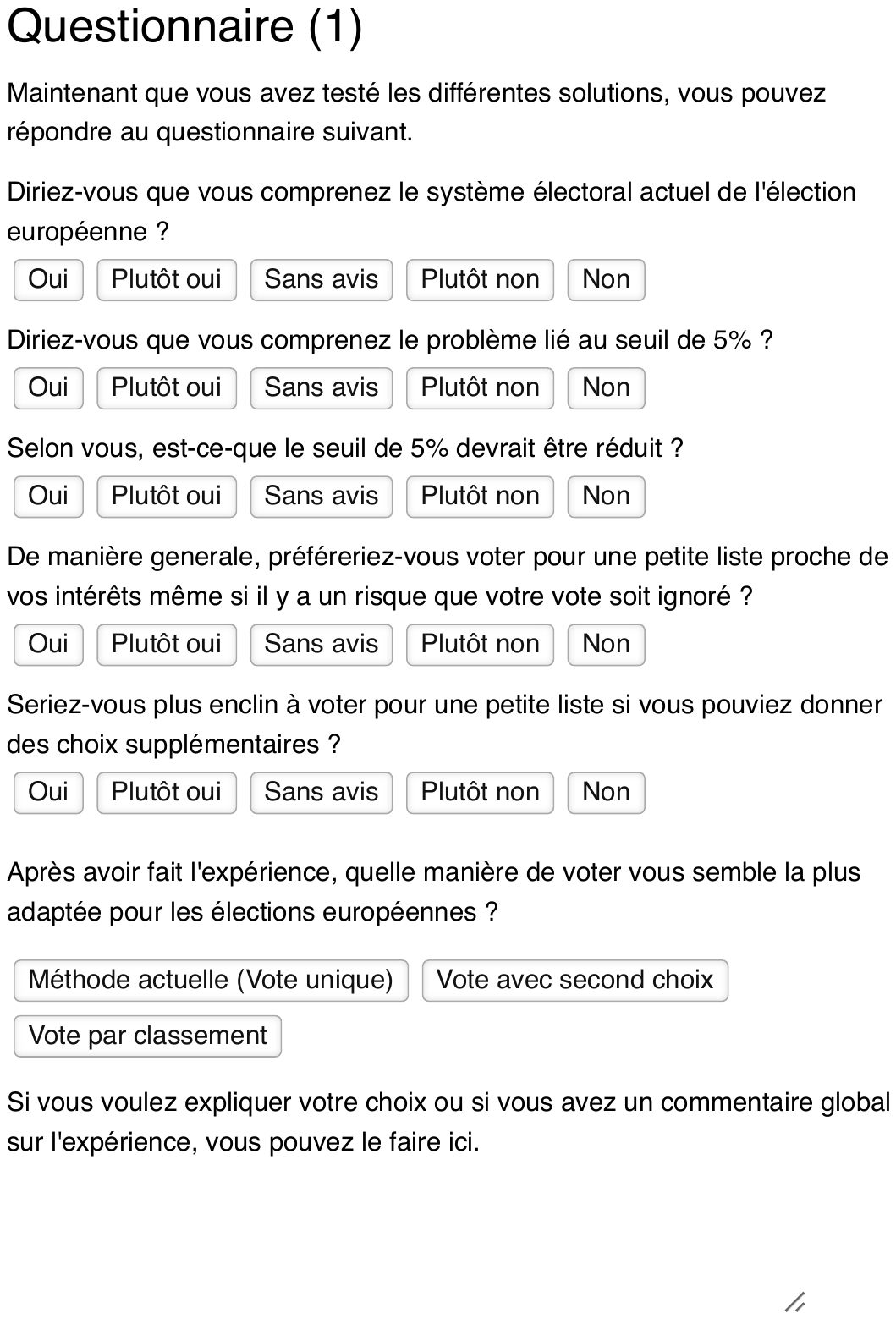}
	\includegraphics[width=0.48\linewidth,frame=0.2pt,trim=1cm 5cm 1cm 1cm,clip]{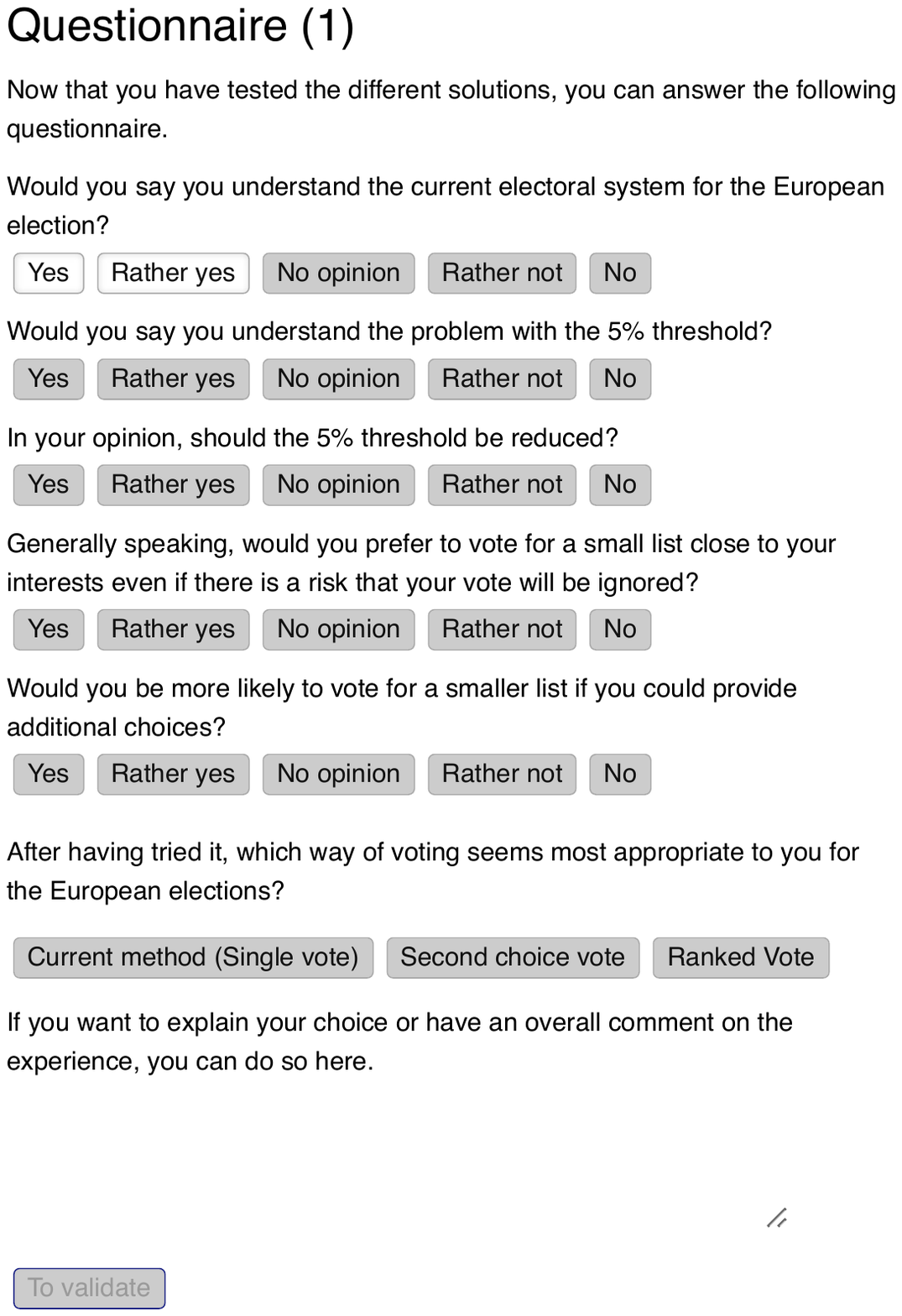}
	\caption{Page 7}
	\label{fig:page7}
\end{figure}

\begin{figure}[h]
	\includegraphics[width=0.48\linewidth,frame=0.2pt,trim=1cm 18cm 1cm 1cm,clip]{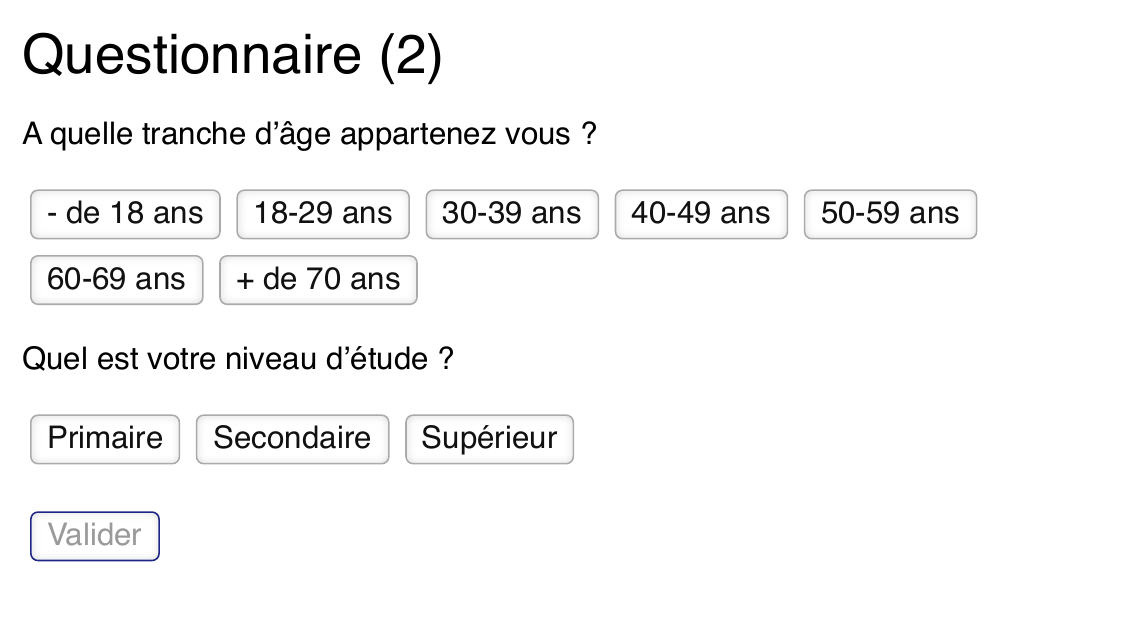}
	\includegraphics[width=0.48\linewidth,frame=0.2pt,trim=1cm 18cm 1cm 1cm,clip]{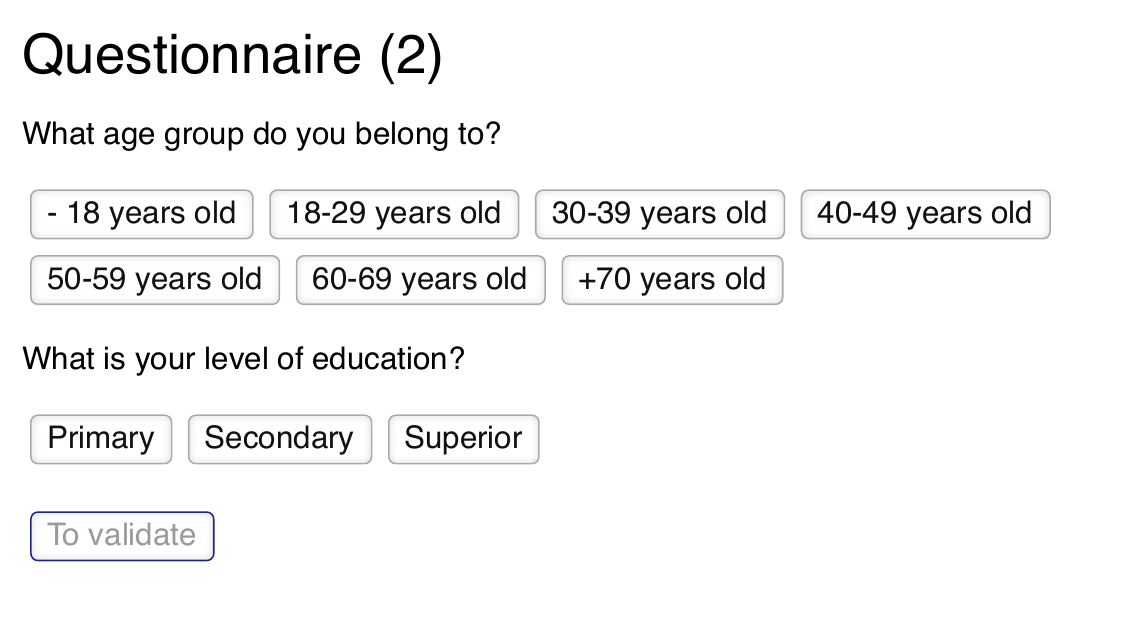}
	\caption{Page 8}
	\label{fig:page8}
\end{figure}

\clearpage

\section{Comparison with the Uninominal System} \label{app:uninominal}

The rule actually used in most proportional elections with thresholds only selects parties whose plurality score is above the threshold. Therefore, seen as a party selection rule, it coincides with DO, but its representation and score functions are determined differently: voter $i$ is represented by its top party  $\best^1_S(i)$ if it belongs to $S$, and by $\emptyset$ otherwise. Thus, the score of a party in $S$ is the number of voters who rank it in top position in their initial vote. This can lead to very different distributions of representation degree, and thus distributions of seats in the parliaments, as a large number of votes may be wasted. For example, in the profile
\begin{align*}
	\text{100: } & a &
	\text{100: } & b &
	\text{99: } & c \succ b 
\end{align*}
with threshold $\tau = 100$, the selected parties are $\{a,b\}$. Under the uninominal system, both parties have 100 supporters, while under DO, party $b$ has almost twice as many supporters as $a$.

We give a second, more realistic example.

\begin{example} Five parties compete: Red (left), Green (ecologist), Pink (center-left), Blue (center-right), Brown (far right)
	\begin{align*}
		\text{8: } & Red \succ Pink \succ Green&
		\text{6: }  & Green \succ Pink \succ Red&
		\text{5: } & Pink \succ Green \succ Red\\
		\text{7: } & Red \succ Green \succ Pink&
		\text{5: }  & Green \succ Red \succ Pink&
		\text{5: } & Pink \succ Red \succ Green\\
		\text{10: }  & Blue \succ Pink&
		\text{15: } & Brown \succ Blue&
		\text{4: } & Pink \succ Blue \succ Green\\
		\text{15: }  & Blue \succ Brown&%
	\end{align*}
	Let $\tau = 15$. With DO we select $\{Red, Blue, Brown\}$ with $\rep(Red) = 36$, $\rep(Blue) = 29$ and $\rep(Brown) = 35$. With the uninominal system however, we have $\rep(Red) = 15$, $\rep(Blue) = 25$ and $\rep(Brown) = 35$, leading to a completely different distribution of seats between the parties. For instance, 
 with the D'Hondt apportionment method and 10 seats in the parlement, we elect four Red, three Blue and three Brown if the party selection rule is DO, while we elect two Red, three Blue and five Brown with the uninominal system.
 
With STV, Green is eliminated first, and the outcome is $\{Red, Pink, Blue, Brown\}$ with $\rep(Red) = 20$, $\rep(Pink) = 20$, $\rep(Blue) = 25$ and $\rep(Brown) = 35$. GP gives the same outcome. The 10 seats would then be distributed as follows: two seats for Red, two for Pink, two for Blue and four for Brown. 
\end{example}

\end{document}